\documentclass[11pt]{llncs}

\usepackage{mathpartir}
\usepackage{amsmath,amssymb}
\usepackage{xcolor}
\usepackage{stmaryrd}
\usepackage{mathbbol}
\usepackage{mathtools}
\usepackage{graphicx}
\usepackage{upgreek}
\usepackage{tikz}
\usepackage{tikzscale}
\usepackage{url}
\usepackage{mathabx,stmaryrd}
\usepackage{fancyvrb}
\usepackage{a4wide}
\usepackage[mathscr]{euscript}

\usepackage{float}
\usepackage{listings}

\newfloat{lstfloat}{htbp}{lop}
\floatname{lstfloat}{Listing}

\newif\ifdraft
\drafttrue

\newcommand{\subst}[2]{\{^{#1}/_{#2}\}}

\newcommand{\semi}{\mathrel{,}}
\newcommand{\comma}{\mathrel{\mbox{\tt ,}}}
\newcommand{\fv}[1]{\mbox{{\em fv}}(#1)}

\renewcommand{\S}{\texttt{S}}

\newcommand{\bigfract}[2]{\frac{^{\textstyle #1}}{_{\textstyle #2}}}
\newcommand{\rulename}[1]{{\small {\sc[#1]}}}

\newcommand{\rulenamex}[1]{\mbox{\tiny [{\sc #1}]}}

\def \mathrule #1#2#3{	\begin{array}{l} 
                       	\rulenamex{#1}
                       	\\ 
                      	 \bigfract{#2}{#3}	
                       	\end{array}
					 }

\newcommand{\C}{\mbox{{\tt C}}}
\newcommand{\state}{\texttt{Q}}
\newcommand{\lolli}{\multimap}
\newcommand{\send}{\rightarrow}
\newcommand{\event}{\gg}
\newcommand{\tostate}{\Rightarrow}
\newcommand{\A}{\texttt{A}}
\newcommand{\State}{\Phi}
\newcommand{\SemEvent}{\Psi}
\newcommand{\Time}{\mathbb{t}}
\newcommand{\Events}[1]{\mathtt{Events}(#1)}

\newcommand{\Solidity}{\textsf{Solidity}}
\newcommand{\smallang}{\mbox{{\sf \emph{Stipula}}}}
\newcommand{\Stipula}{\mbox{{\sf \emph{Stipula}}}}

\newcommand{\program}{\mathcal{P}}

\newcommand{\contract}{\mathbb{C}}

\newcommand{\sol}{\mathbb{S}}

\newcommand{\zero}{\mbox{\raisebox{-.6ex}{{\tt -\!\!\!-}}}}

\newcommand{\sem}[1]{\llbracket#1\rrbracket}

\newcommand{\f}{\texttt{f}}
\newcommand{\g}{\texttt{g}}
\newcommand{\afield}{\texttt{h}}
\newcommand{\bfield}{\texttt{x}}

\newcommand{\lred}[1]{\mathrel{\stackrel{#1}{\longrightarrow}}}
\newcommand{\Lred}[1]{\mathrel{\stackrel{#1}{\Longrightarrow}}}

\newcommand{\nored}[2]{#1\comma #2 \mathrel{\nrightarrow}}
\newcommand{\xred}[1]
{ \setbox0=\hbox{$\, {}^{#1}\, $}
  \setbox1=\hbox{$\longrightarrow$}
  \loop\setbox1=\hbox{$-$\kern-0.3em\unhbox1}\ifdim\wd1<\wd0\repeat
  \hbox{$\ \mathop{\box1}\limits^{#1} \ $}
}

\newcommand{\wt}[1]{\vect{#1}}

\newcommand{\eqdef}{\stackrel{\textsf{\tiny def}}{=}}

\newcommand{\vect}[1]{\overline{#1}}

\newcommand{\marginnote}[2]{%
  {\makebox[0pt]{\color{magenta}$^\bigstar$}}%
  \marginpar{\parbox{2cm}{\flushright\tiny\sffamily\textbf{#1}: #2}}%
}

\ifdraft
\newcommand{\Silvia}[1]{\marginnote{Silvia}{\color{red}#1}}
\newcommand{\Cosimo}[1]{\marginnote{Cosimo}{\color{blue}#1}}
\newcommand{\Giovanni}[1]{\marginnote{Giovanni}{\color{green}#1}}
\newcommand{\giovanni}[1]{GIOVANNI: \textbf{#1}}
\newcommand{\Reviewer}[2]{\marginnote{Revw #1}{\color{black}#2}}
\else
\newcommand{\Silvia}[1]{}
\newcommand{\Cosimo}[1]{}
\newcommand{\Giovanni}[1]{}
\newcommand{\giovanni}[1]{}
\newcommand{\Reviewer}[2]{}

\fi

\begin{document}

\title{Pacta sunt servanda:  legal contracts in {\smallang}}
\author{Silvia Crafa\inst{1} \and Cosimo Laneve\inst{2} \and Giovanni Sartor\inst{3}}

\institute{University of Padova \and University of Bologna -- Inria FOCUS \and
University of Bologna -- European University Institute
}



\maketitle

\begin{abstract} 
There is a growing interest in running legal contracts on digital systems, 
at the same time, it is important to understand to what extent software contracts may capture legal content. 
%
We then undertake a foundational study of 
legal contracts and we distill four main features: agreement, permissions, violations 
and obligations. 
We therefore design {\Stipula}, a domain specific language that assists lawyers 
in programming legal contracts through specific patterns. The language is based 
on a small set of abstractions that correspond to common patterns in legal contracts, and that are amenable to be executed either on centralized or on distributed systems.
{\Stipula} comes with a formal semantics and an observational equivalence,
that provide for a clear account of the contracts' behaviour. 
%
%
The expressive power of the language is illustrated by a set of examples that correspond to template contracts that are often used in practice.
\end{abstract}




\section{Introduction}
\label{sec:introduction}

A legal contract is ``an agreement which is intended to give rise to a binding
legal relationship or to have some other legal effect'' \cite{CommonFrameReference2009Pr}. 
The parties are in principle free to determine 
the content of their contracts (\emph{party autonomy/freedom of contract}):  
the law recognizes  their intention to achieve the agreed outcomes and 
secures the enforcement of such outcomes (\emph{legally binding effect}). 
A contract produces the intended effects, declared by the parties, only if it is legally valid: 
the law may deny validity to certain clauses (\emph{e.g.}~excessive interests rate) 
and/or may establish additional effects that were not stated by the 
parties (\emph{e.g.} consumer's power to withdraw from an online sale, warranties, etc.). 
The intervention of the law is particularly significant when the contractor (usually the weaker party, such as the worker in an employment contract or the consumer in an online purchase) agrees  without having awareness of all clauses in the contract, nor having the ability to negotiate them, due to the existing unbalance of power. 

The assimilation of software contracts to legally binding contracts, or rather the double nature of digital contracts as computational mechanisms and as legal contracts, raises both legal and technological issues. 
Blockchain-based smart contracts have been advocated for digitally encoding legal contracts, so that the execution and enforcement of contractual conditions may occur automatically, without human intervention. However, 
we will mostly refer to \emph{software contracts} as ``digital legal contracts'', stressing 
the fact that most of the benefits of digitally encoding legal contracts come from 
the precise definition and automatic execution of a piece of programmable software, not necessarily operating over a blockchain.

In this paper, after discussing the main issues raised by the literature, we propose
a  technology that may contribute to addressing them.
The overall aim is to facilitate the transparency of  software contracts, as well 
as the mapping of  computational operations 
into legal-institutional  outcomes,  thus limiting or mitigating the 
problems concerning the implementation of software contracts and preventing
disputes between the parties. 
%
%

%
%
Specifically, we put forward {\smallang}, a new domain specific language for 
the creation of  software legal contracts.
{\smallang} is pivoted on a small number of abstractions that are useful to capture the 
distinctive elements of legal contracts, that is \emph{permissions}, \emph{prohibitions}, 
\emph{obligations}, fungible and non fungible \emph{assets} exchanges, 
risk (viz.~\emph{alea}),  \emph{escrows} and \emph{securities}.
All these normative elements are expressed by a strictly regimented behaviour in
legal contracts:
permissions and empowerments correspond to the possibility of performing an action
at a certain stage, 
prohibitions correspond to the interdiction of doing an action, while
obligations are recast into commitments that are checked at a specific time limit. 
Moreover, the set of normative elements changes over time according to the 
actions have been done (or not).
To model these changes, 
{\smallang} commits to a \emph{state-aware programming} style, inspired by the state machine 
pattern widely used in the programming language practice. 
This technique allows one to enforce the intended 
behaviour by prohibiting, for instance,  the invocation of a function before 
another specific function is called.

A second distinctive feature of {\smallang} is the \emph{event} primitive, 
a programming abstraction that is used to issue an obligation and schedule a future 
statement that automatically executes a corresponding penalty, 
if the obligation is not met. This allows one to implement legal obligations and 
commitments in terms of the future execution of a (state-aware) computation at 
a specific point in time. 

A third peculiarity of {\smallang} is the \emph{agreement operator}, which
marks that the contract's parties have reached a
consensus on the contractual arrangement they want to create.
In legal contracts, this phase corresponds to the 
subscription of the contract, where there are
parties that are going to set the 
contractual conditions and others that accept them.
%
Technically, the operation is
a multiparty synchronization, which may be implemented 
by ad-hoc protocols in  many kinds of distributed systems.
%
%
%
%
%
%

The fourth key feature regards \emph{assets}, which are first-class 
linear concepts in {\smallang}.
In particular, {\smallang} has an explicit, and thus conscious, management of linear resources,  
such as currency-like values and tokens.   
The transfer of such resources must preserve the total supply: the sender of the asset 
must always relinquish the control of the transferred asset. 
These assets 
are necessary in legal contracts: 
currency is required for payment but also for escrows, and tokens, both fungible 
and non-fungible, are useful to model securities and provide a digital 
handle on a physical good  (possibly equipped with and IoT mechanism). 
The design choice of explicitly marking asset movements with 
an ad hoc syntax in {\smallang}
promotes a safer, asset-aware, programming discipline that reduces the risk of the 
so-called double spending, the accidental loss or the locked-in assets. 
This is particularly useful for legal contracts running over a blockchain, where many kinds of assets and tokens are pervasive in the most successful blockchain applications and their induced economy.
Therefore,
{\smallang} uses a simple and powerful 
core of primitives that support the writing of digital legal contracts even by non 
ICT experts.

In Section~\ref{sec:legalcontracts} we give an interdisciplinary discussion about smart legal contracts, focusing on the interlace between the  digital and legal 
elements required by the digitalization of juridical acts.
The syntax of {\smallang} is formally defined in Section~\ref{sec:smartlegalcontracts}
where a simple example -- the Bike-Rental contract -- is used to describe 
the concepts. {\smallang} semantics is defined in Section~\ref{sec:semantics}, 
sticking to
an operational approach that specifies the runtime behaviour of legal contracts by means 
of transitions. 
In Section~\ref{sec:laws}, following a standard technique in concurrency theory~\cite{MilnerBook}, we develop an observational equivalence that provides for an equational theory of smart legal contracts. 
The equivalence is based on a notion of \emph{bisimulation} that
equates contracts
differing for hidden elements, such as names of states, and singles out conditions
for identifying contracts that send two assets in different order.
The study of the implementation of the distinctive elements of {\smallang}, namely agreements, assets and events, either on a blockchain system, using a {\Solidity}-like target language, or not, is 
undertaken in Section~\ref{sec:expressivepower}. The specification in {\smallang} of a set of archetypal acts, ranging from renting to (digital) licenses and to bets, is 
reported in Section~\ref{sec:examples}

We end our contribution by discussing the related work in Section~\ref{sec:relatedworks}
and delivering our final remarks in Section~\ref{sec:conclusions}. 

%
%
%
%
%
%



\section{Smart Legal contracts}
\label{sec:legalcontracts}

A substantial debate has taken place on whether the parties' decision to execute a 
smart contract having certain computational effects may count as legal 
contract establishing  corresponding legal effect, 
\emph{e.g.}~\cite{Finocchiaro2020,Sartor2018,Lemme2019,Radin2017Th}. 
A simple answer to this question comes from the
principle of ``\emph{freedom of form}'' in contracts, which is shared by modern legal 
systems: \emph{parties are free to express their agreement using the language and medium 
they prefer, including a programming language}. Therefore, by this principle, 
smart contracts may count as legal contracts. 

However, the problem is
whether smart contracts preserve the essential elements of legal ones. 
%
%
In this
respect, a legal contract is meant to bring about the \emph{institutional effects}
intended by the parties, that is establishing new  obligations,  rights, powers 
and liabilities between them or to transfer rights (such as rights to property) 
from one party to the other. These institutional effects are 
guaranteed by the  possibility of activating judicial enforcements. 
That is, each party may start a lawsuit if he believes that the other party has failed 
to comply with the contract. In this case, the judge will have to interpret the contract, 
ascertain the facts of the case, and determine whether there has indeed been 
a contractual violation. Accordingly, the defaulting party may be enjoined to comply 
or  pay damages.
%
%
%
%
Whether and to what extent we may consider that a smart contract produces judicially 
enforceable legal effects is more debatable, given that smart contract modify especially in absence of international 
technical standards and transnational legal frameworks.

We recognise that no easy and comprehensive solution is yet at hand for the 
issue we have just mentioned. However, we believe that it can be  at least mitigated 
if a strict, and understandable mapping is established between executable instructions and 
institutional-normative effects.  
To this aim, we observe that 
the lifecycle of a legal contract goes through a number of phases: 
(\emph{a}) formation and negotiation, (\emph{b}) contract storage/notarizing,
(\emph{c}) performance, enforcement and monitoring, (\emph{d}) possible modification
(\emph{e}) dispute resolution, and (\emph{f}) termination.
Software-based solutions can be valuable in all these phases, additionally, the specific features of blockchain-based implementations make them convenient in some of them. 
For example, the negotiation of contractual conditions could benefit from software-mediated interactions, but might require a degree of privacy that conflicts with that of a public blockchain, 
which naturally runs on transnational 
infrastructures, thus crossing several, possibly different, legal systems and jurisdictions (\emph{c.f.}~General Data Protection Regulation is valid only in Europe).
Similarly the dispute resolution can take advantage of online services with flexible interfaces, since it can hardly be fully programmed nor ported on-chain.
Exceptional behaviors, such as mutual or unilateral dissent, contract termination or contract modifications matching a change in the parties'~will, are also very problematic and require a flexible programming style that admits suitable patterns to amend the behaviour of a 
running software.

On the other hand, software-based solutions, possibly based on the blockchain, are perfectly suited to phases (\emph{b}) and (\emph{c}). Both the content of the legal contract and the expression of agreement of the parties can be digitally notarized, and code can be used to express the contractual clauses and automatically enforce 
them through the runtime execution. Additionally, encoding legal contracts as software has the advantage of enabling the usage of verification (formal) methods to ensure the correctness of the software execution.

%
Nevertheless, if smart contracts are legally binding, then it is necessary to ensure 
that the parties are fully aware of the computational effects of their code. 
Only in this case, there may be a genuine  agreement over the content of the contract. 
Thus transparency and some degree of readability by contractors that have no or little 
computer expertise becomes a key requirement.
Relating to this point, 
we acknowledge that similar problems  also exist in natural language contracts,
which are usually signed without the parties being  aware of all of their  clauses
and judicial enforcement being too costly or complex to be practicable. 
Usually, when concluding online purchases of  goods or services,  most consumers  
just click the ``accept'' button, without even trying to read the clauses (which are often lengthy and full of legal jargon). 


We believe that {\smallang} provides some important advantages as a language 
for specifying smart legal contracts. Its primitives are concise and abstract enough to be
easily accessible to lawyers. Its formal operational semantics 
promises that the 
execution of contracts does not lead to unexpected behaviours and is amenable to
automatic verification.
%
%
%
To mention a distinctive feature of {\smallang}, the  
{\tt agreement} primitive helps to deal with some legal issues, since it 
clearly identifies the moment when (some) legal effects are triggered and the 
parties who are involved. 
For instance, the set of parties involved in the agreement might include an 
Authority that is charged to monitor contextual 
constraints, such as obligations 
of diligent storage and care, or the obligations of using goods only as intended, 
taking care of litigations and dispute resolution. Moreover, untrusted players 
involved in a bet contract can rely on the agreement to explicitly define  
the data source providing the outcome of the aleatory value associated to the bet.

\section{The \smallang\, language}
\label{sec:smartlegalcontracts}

\begin{table}[t]
{\small
\begin{tabular}{cc}
\\
\hspace{1cm} &
\begin{lstlisting}[mathescape,basicstyle=\ttfamily]
F $::=$  _ $|$ @$\state$ $\A :$ f($\vect {\tt z}$)[$\vect {\tt y}$] (B){$\,$S W$\,$}$\,\tostate\,$@$\state'\; $ F
S $::=$  _ $|$ E $\send$ x S $|$  E $\lolli$ h,h$'$ S  $|$ $\texttt{E}$ $\send$ $\A$ S $|$ E $\lolli$ h,$\A$ S  $|$ (B){$\,$S$\,$} S
W $::=$  _  $|$ E$\,\event\,$@$\state\,${$\,$S$\,$}$\,\tostate\,$@$\state'\; \;$  W 
E $::=$  now  $|$  Real  $|$  String  $|$  B  $|$  X  $|$  E op E
B $::=$  true  $|$  false  $|$  E rop E  $|$  B $\&\&$ B    $|$  B || B  $|$  ~B
\end{lstlisting}
\\
\\
\end{tabular}
}
\caption{\label{tab:syntax} Syntax of {\smallang}.}
\end{table}

{\smallang} features a minimal set of primitives that are primary in legal 
contracts, such as agreements, 
field updates, conditional behaviour, timed events, value and asset transfer, 
functions and states like Finite State Machines (FSMs).

We use countable sets of: \emph{contract names}, ranged over by $\C$, $\C'$, $\cdots$; names of externally owned accounts, called~\emph{parties},
ranged over by $A$, $A'$, $\cdots$;  \emph{function names}  ranged over $\f$, $\g$, $\cdots$.
Parties represent the users involved in the 
contract, \emph{i.e.}~authenticated users in centralized systems (or 
addresses in blockchain systems).
\emph{Assets} and generic 
contract's \emph{fields} are syntactically set apart since they have different semantics. 
Then we assume a countable set of \emph{asset names}, ranged over by $\afield$, $\afield'$, 
$\cdots$, and a set of \emph{field names}, ranged over $\bfield$, $\bfield'$, 
$\cdots$. We reserve \texttt{z}, \texttt{z}$'$, \texttt{y}, \texttt{y}$'$, for  function parameters and $\A, \A', \cdots$ for parameters that are parties. 
Finally, we will use $\state$, $\state'$, $\cdots$, to range over contract states.
%
To simplify the syntax, we often use the vector notation ${\vect x}$ to denote possibly empty sequences of elements.
A smart legal contract in {\smallang} is written

{\small
\begin{lstlisting}[mathescape,basicstyle=\ttfamily]
    stipula C {
        assets ${\vect {\tt h}}$  
        fields $\vect{\tt x}$ 
 
        agreement($\vect{\tt A}$)(${\vect{\tt x}'}$){   $\quad\quad$// $\vect{\tt x}'\subseteq {\vect{\tt x}}$
             $\vect{{\tt A}_1}$ : $\vect{\tt x}_1$
             $\cdots$
             $\vect{{\tt A}_n}$ : $\vect{\tt x}_n$   $\quad\quad$ // $\bigcup_{i \in 1..n} \vect{{\tt A}_i} \subseteq \vect{{\tt A}}      \qquad       
                       \bigcup_{i \in 1..n} \vect{{\tt x}_i} \subseteq \vect{{\tt x}'}
	                   \qquad	\bigcap_{i \in 1..n} \vect{{\tt x}_i} = \varnothing$
        } $\tostate$ @$\state$  

        F
    }
\end{lstlisting}
}

\noindent
where \texttt{C} identifies the contract; its body contains assets and fields (without any
typed information: {\smallang} is type-free), the \emph{agreement code},
where $\bigcup_{i \in 1..n} \vect{{\tt x}_i}$ is a partition of $\vect{{\tt x}'}$ while the sequences of parties $\vect{{\tt A}_i}$ are subsets of the full list $\vect{\tt A}$ of parties involved in the contract \texttt{C}. Finally, {\tt F} is a sequence 
of functions, written according to the syntax given in Table~\ref{tab:syntax}.

The definition of \emph{who} is going to participate to the contract and
\emph{what} are the terms of the contract in an explicit abstraction 
-- the \emph{agreement} -- is another distinctive feature of {\smallang}. 
Technically, the agreement is the 
\emph{constructor} of the contract, that specifies the parameters $\vect {\tt x'}$, i.e., the terms of the contract, and the set $\vect {\tt A}$ of involved parties. Moreover, the code specifies who must agree on the initial value of contract's fields and
 the initial state of the contract.
Observe that no asset can be set during the agreement. 
It is assumed that fields, assets and parties' names do not contain duplicates. 

The dichotomy between assets and fields is a key design choice of {\smallang}.
Indeed, the relevance of first-class resources, \emph{i.e.}~linear values that 
cannot be copied nor dismissed, is widely acknowledged in several programming languages
(such as {\tt Rust} or the smart contract languages~{\sf Move}~\cite{Move} 
and {\sf Nomos}~\cite{Nomos})
to support a safer asset-aware contract programming.
Legal contracts also manage assets, such as money and tokens granting a digital 
access to (possibly physical) goods or 
services. Henceforth the decision to syntactically highlight the differences 
between operations on values and on assets. 

Functions \texttt{F} and their bodies are written according to the syntax in Table \ref{tab:syntax}.
The function's syntax highlights the constraint that only the party {\A} can invoke the function \texttt{f}, and only when the contract is in state \texttt{\state}.
%
Function's parameters are split in two lists: the \emph{formal parameters} 
$\vect {\tt z}$ in brackets and the \emph{asset parameters}
$\vect {\tt y}$ in square brackets. 
The \emph{precondition} \texttt{(B)} is a predicate on parameters of the functions and 
fields and assets of the 
contract that constrains the execution of the body of \texttt{f}. 
Finally the \emph{body} \texttt{\{S W\}\,$\tostate$\,@\state}$'$ specifies the 
\emph{statement part} $\texttt{S}$, the \emph{event part} \texttt{W}, and the state 
$\state'$ of the contract when the function execution terminates.
Function's parameters are assumed without duplicates, and 
empty lists of (asset) parameters are shortened by omitting empty parenthesis.
Additionally we assume that function parameters $\vect{\tt z}$ and
$\vect{\tt y}$ do not occur in \texttt{W} to enforce that it is correctly executed outside
the scope of the function. 

\emph{Statements} \texttt{S} include the empty statement $\zero$ and different types of 
assignments, followed by a continuation. Assignments use the two symbols $\lolli$ and 
$\send$ to differentiate updates of assets and of fields, respectively. 
The syntax of the two operators is taken from~\cite{Nomos}.
%
Assignments can be either \emph{local}, that is referring to local fields or local assets, denoted by 
\texttt{E}~$\send$~\texttt{x} and \texttt{E}~$\lolli$~\texttt{h,h$'$}, respectively, 
or they can be \emph{remote}, denoted by $\texttt{E}$~$\send$~\texttt{\A} and 
\texttt{E}~$\lolli$~\texttt{h,\A}, defining the sending of a value and an asset, 
respectively, to the address \A. 
Asset assignments are ternary operations: the meaning of 
\texttt{E}~$\lolli$~\texttt{h,h$'$} is that the value of \texttt{E} is subtracted to the asset
\texttt{h} and added to the asset \texttt{h$'$} -- \emph{resources stored in assets can be 
moved but cannot be destroyed}. 
The operational semantics will ensure that asset assignments can at most drain an asset, preventing assets with negative values.
In the rest of the paper we will always abbreviate assignments such as 
\texttt{h}~$\lolli$~\texttt{h,h$'$} and \texttt{h}~$\lolli$~\texttt{h,\A$'$} (which are 
very usual, indeed) into \texttt{h}~$\lolli$~\texttt{h$'$} and \texttt{h}~$\lolli$~\texttt{\A$'$},
respectively. 
%
Statements also include 
\emph{conditionals} \texttt{(B)\{\,S\,\} S$'$} that executes \texttt{S} if the 
predicate \texttt{B} is true and continues as \texttt{S$'$}.

\emph{Events} \texttt{W} are sequences of \emph{timed continuations} that schedule some code for future execution. More precisely, the term 
\texttt{E} $\event$~\texttt{@\state\,\{S\}}\,$\tostate$\,\texttt{@\state$'$} schedules an execution that is triggered  at a time $\Time$ that is the value of \texttt{E}.
When triggered, the continuation 
\texttt{S} will be executed only if the contract's state is {\state}. At the end of the execution of \texttt{S}, the contract transits to {\state$'$}. 
That is, in {\smallang}, the programming abstractions of FSMs
are used to schedule the future execution of a (state-aware) computation at a specific 
point in time. We will show that this notion of events is pivotal in the encoding of legal obligations and commitments.

\emph{Expressions} \texttt{E} can be {\tt Real} (written with the fixed-point notation),
{\tt String} (written between ``{\tt "}''), names of assets, fields and parameters, generically ranged over by \texttt{X}. Expressions also include boolean values and boolean expressions
{\tt B}. The keyword {\tt now}, expressed as a real, 
stores the present time (when the code is executed).
We will range over constant \emph{values} with 
the names $u, v, \cdots$; they also include asset constants like tokens.
%
We will be a little liberal with values and operations, 
generically denoted by \texttt{E op E}: they include standard arithmetic operations 
and operations on tokens, \emph{e.g.} \texttt{use\_once(token)} generates a 
usage-code providing a single access to the service or the good associated 
to the \texttt{token} asset. The operation $\Time$\texttt{ + n} sums \texttt{n} seconds 
to the time 
value $\Time$. 
\emph{Boolean expressions} \texttt{B} are the standard ones, where the operations 
\texttt{rop} are the relational operations (\texttt{==}, \texttt{>=}, etc.) 
The set of names occurring in \texttt{E} will be noted by $\fv{\mbox{\tt E}}$.

A {\smallang} program $\program$ is a sequence of smart legal contracts definitions. 
The contracts are inactive as long as no group of addresses has interest to run them, 
by invoking the agreement code. 
We remark that in {\smallang} the code of a contract cannot invoke another contract:
we postpone to future work the study of language extensions allowing cross references 
between legal contracts using inheritance and composition. 
Therefore, at the moment, since legal contracts are independent, there is no loss 
of generality in considering a program to be composed by a single smart legal 
contract definition.

\paragraph{Example}
\label{ex.BikeRental}
Consider the following simple contract for renting bikes:

{\footnotesize
\begin{lstlisting}[numbers=left,numberstyle=\tiny,mathescape,basicstyle=\ttfamily,caption={The rent for free contract},captionpos=b] 
stipula Bike_Rental {  
  assets wallet
  fields cost, rent_time, use_code

  agreement (Lender, Borrower)(cost, rent_time) {
     Lender, Borrower : cost, rent_time
  } $\tostate$ @Inactive

  @Inactive Lender : offer (z) {
     z $\send$ use_code  
  } $\tostate$ @Proposal

  @Proposal Borrower : accept [y] 
    (y == cost) {
      y $\lolli$ wallet
      use_code $\send$ Borrower
      now + rent_time  $\event$ 
          @Using {              //end-of-time usage
              "End_Reached" $\send$ Borrower 
              wallet $\lolli$  Lender 
          } $\tostate$ @End
  } $\tostate$ @Using

  @Using Borrower : end {
    wallet $\lolli$ Lender 
  } $\tostate$ @End
}
\end{lstlisting}
}

The agreement code specifies that there are two parties -- the \texttt{Lender} and \texttt{Borrower} -- and that the \texttt{Lender} sets the values for the time of usage (\texttt{rent\_time})
and the cost, while the \texttt{Borrower} has to agree on such values. Then \texttt{Lender}
sends a use-code that is stored in the contract (in the \texttt{use\_code} field) 
and is not accessible to \texttt{Borrower} till he pays for the usage. The transition from \texttt{Inactive} to \texttt{Proposal} at the end of the \texttt{offer} function enables
the \texttt{Borrower} to pay for the usage -- function \texttt{accept}  -- 
that 
takes in input an asset {\tt y} and moves it into {\tt wallet} with 
{\tt y $\lolli$ wallet} at line 15 
(this is a shortening for
{\tt wallet $\lolli$ wallet, Borrower}). Then the contract sends the use-code
to the borrower  (line 16), which now can use the bike till the time limit. 
This constraint is expressed by the 
event in lines 17-21. We have two remarks: first, the payment is not made to \texttt{Lender}
but the asset is stored in the contract (in \texttt{wallet}); second, the event will be
triggered when the time expires. In this case a message to \texttt{Borrower} is sent,
the payment is transferred to \texttt{Lender}, which will change bike's use code so that the bike will be locked at the next \texttt{Borrower}'s stop.
The function \texttt{end} can be invoked by \texttt{Borrower} to terminate the renting
before time expires. The legal issues involved in a rent contract will be discussed in Section~\ref{sec:examples}.

\section{Semantics}
\label{sec:semantics}

The meaning of {\smallang} primitives is defined in an operational way 
by means of a transition 
relation.
Let $\C(\State\semi\ell\semi \Sigma \semi\SemEvent)$ be a \emph{runtime contract} where 
\begin{itemize}
\item
$\C$ is the contract name;
\item 
$\State$ is the current state of the contract: it is either $\zero$ (for no state) or a contract state \state;
\item
$\ell$ is a mapping from fields and assets to values;
\item
$\Sigma$ is  a possible residual of a function body or of an event handler, \emph{i.e.}~$\Sigma$ is either 
$\zero$ or a term  $\texttt{S}\ \texttt{W}\tostate \texttt{@}\state$;
\item
$\SemEvent$ is a (possibly empty) multiset of \emph{pending events} that have been already scheduled 
for future execution but not yet triggered. We let $\SemEvent$ be $\zero$ when there are no pending events, otherwise $\SemEvent=\texttt{W}_1 ~|~ \ldots ~|~ \texttt{W}_n$ 
such that each \texttt{W}$_i$ is a single event expression (not a sequence), and its \emph{time guard} is an expression that has already been evaluated into a time value $\Time_i$. 
\end{itemize}
Runtime contracts are ranged over by $\contract$, $\contract'$, $\cdots$.
A \emph{configuration}, ranged over by $\sol$, $\sol'$, $\cdots$, is a pair 
$\contract {\comma} \Time$, where $\Time$ is a global clock. 
As anticipated, there is no loss of generality in considering programs made of 
a single running contract; a remark about the extension to configurations 
containing several contracts is reported at the end of the section.

The transition relation of a {\smallang} program $\program$ 
is $\sol \lred{\mu}_{\program} \sol'$, where 
$$ \mu \quad ::= \quad \_  \quad |\quad  (\vect A,\, {\vect A_1} :  \vect v_1 , \cdots, 
 {\vect A_n} : \vect v_n) \quad |\quad  
                 A : \texttt{f}(\wt{u})[\wt{v}] \quad |\quad
                 v \send A \quad |\quad  v \lolli A 
$$
that is the label $\mu$ is either empty, or denotes an initial agreement, or a function call, or a value send, or an asset transfer.
In the following we will always omit the 
index $\program$ because it is considered implicit.
The formal definition of $\sol \lred{\mu} \sol'$ is given in Table~\ref{tab:sem1}, 
using the following auxiliary predicates and
functions:
\begin{itemize}
\item
$\sem{\texttt{E}}_{\ell}$ is a function that returns the value of \texttt{E}
in the memory $\ell$.
We omit the definition. 
\item
$\sem{\texttt{W}}_\ell$ is the multiset of scheduled events obtained from the 
sequence \texttt{W} by replacing every time expression by its value in the time guards. 
That is $\sem{\zero}_\ell = \zero$ \ and 
$\sem{ \texttt{E}\event\texttt{@}\state\{\texttt{S}\}\tostate\texttt{@}\state' 
\ \ \texttt{W}'}_\ell = \sem{ \texttt{E}}_\ell 
\event \texttt{@}\state\{\texttt{S}\}\tostate\texttt{@}\state' ~|~ \sem{\texttt{W}'}_\ell$.
\item 
Let $\SemEvent$ be a multiset of pending events, and $\Time$ a time value, then 
the predicate $\nored{\SemEvent}{\Time}$ is \emph{true} whenever 
$\SemEvent = \Time_1 \event \texttt{@}\state_1 \{ \texttt{S}_1 \} \tostate
\texttt{@}\state_1' ~|~ \cdots ~|~  \Time_n \event \texttt{@}\state_n \{ 
\texttt{S}_1 \} \tostate
\texttt{@}\state_n'$ and, for every $1 \leq i \leq n$, $\Time_i \neq \Time$, \emph{false}
otherwise.
\end{itemize}

%
\begin{table*}[t]
{\small
\[
\begin{array}{c}

%

\mathrule{Agree}{
	\begin{array}{c}
	\texttt{agreement(} \vect{\tt A} \texttt{)(}\vect{\tt x}\texttt{)\{ } 
	\vect{\tt A}_1 : \vect{\tt x}_1  \ \cdots\   
	\vect{\tt A}_n : \vect{\tt x}_n \ \} 
	\tostate \texttt{@}\state \ \in \; \C
	\end{array}
	}{
	\C(\zero \semi \emptyset \semi \zero\semi \zero)\comma\Time 
	\lred{(\vect A,\, {\vect A_i}{:} \vect v_i \ ^{i=1...n})} 
	\C(\state \semi 
	[{\vect{\tt A}} \mapsto {\vect A}, \vect{\tt x}_i~\mapsto~ \vect v_i \ ^{i=1...n}] 
	\semi \zero \semi \zero)\comma \Time 
	}
\\[.3cm]
%

\mathrule{Function}{
	\begin{array}{c}
	\texttt{@}\state~\A~\texttt{:~f(} \wt{\mbox{{\tt z}}} \texttt{)[} \wt{\mbox{{\tt y}}} \texttt{] (B) \{ S~W \}} \tostate \texttt{@}\state' \in \C
	\\
	\nored{\SemEvent}{\Time}
	\\
	\ell(\A) = A
	\quad
	\ell' = \ell[ \wt{\mbox{{\tt z}}} \mapsto \wt{u}, \wt{\mbox{{\tt y}}} \mapsto \wt{v} ]
	\quad
	\sem{\texttt{B}}_{\ell'} = {\it true}
	\end{array}
	}{
	\C(\state \semi \ell \semi \zero \semi \SemEvent) \comma \Time 
	    \lred{A : \texttt{f}(\wt{u})[\wt{v}]}
	\C(\state ; \ell' \semi \texttt{S}\,\texttt{W} \tostate \texttt{@}\state' 
	\semi \SemEvent) \comma \Time
	}
	\quad
\mathrule{State Change}
    {\sem{\texttt{W}\subst{\Time}{\texttt{now}}}_{\ell} = \SemEvent'}
    {
	\C(\state \semi \ell \semi \zero \,\texttt{W}\tostate \texttt{@}\state' \semi \SemEvent) 
	\comma \Time 
	\lred{}
	\C(\state' \semi \ell \semi \zero \semi \SemEvent' ~|~\SemEvent) 
	\comma \Time
	} 

\\[1cm]
\mathrule{Event Match}{
	\begin{array}{c}
	\SemEvent = \Time \event \texttt{@}\state~\{~\texttt{S}~\} \tostate \texttt{@}\state' ~|~ \SemEvent'
	\end{array}
	}{
	\C(\state \semi \ell \semi \zero \semi \SemEvent) \comma \Time \lred{}
	\C(\state \semi \ell \semi \texttt{S}\tostate \texttt{@}\state' \semi \SemEvent') 
	\comma \Time  } 
\qquad
\mathrule{Tick}{
	\nored{\SemEvent}{\Time}
	}{
	\C(\state \semi \ell \semi \zero \semi \SemEvent) \comma \Time \lred{} 
	\C(\state \semi \ell \semi \zero \semi \SemEvent) \comma \Time + 1
	}
\\
\\
\mathrule{Value\_Send}{
	\sem{\texttt{E}}_{\ell} = v \quad \ell(\A) = A
	}{
	\C(\state \semi \ell \semi \texttt{E} \send \A~ \Sigma \semi \SemEvent )
	\lred{v \send A}
	\C(\state \semi \ell \semi \Sigma \semi \SemEvent )
	} 
\quad
\mathrule{Asset\_Send}{
	\sem{\texttt{E}}_{\ell} = v \quad \ell(\texttt{h}) \geq v  \quad \ell(\A) = A
	}{
	\C(\state \semi \ell \semi \texttt{E} \lolli \texttt{h,}\A~\Sigma \semi \SemEvent )
	\lred{v \lolli A}
	\C(\state \semi \ell[\texttt{h} \mapsto \ell(\texttt{h}) - v] \semi \Sigma \semi
	\SemEvent )
	} 
\\[.6cm]
\mathrule{Field\_Update}{
	\sem{\texttt{E}}_{\ell} = v
	}{
	\C(\state \semi \ell \semi \texttt{E} \send \texttt{x}~\Sigma \semi \SemEvent )
	\lred{}
	\C(\state \semi \ell[\texttt{x} \mapsto v] \semi \Sigma \semi \SemEvent )
	} 
\quad
\mathrule{Asset\_Update}{
	\sem{\texttt{E}}_{\ell} = v \qquad \ell(\texttt{h}) \geq v
	}{
	\begin{array}{l}
	\C(\state \semi \ell \semi \texttt{E} \lolli \texttt{h,h}'~\Sigma \semi \SemEvent )
	\\
	\qquad
	\lred{}
	\C(\state \semi \ell[\texttt{h} \mapsto \ell(\texttt{h})-v, \texttt{h}' \mapsto \ell(\texttt{h}') + v] \semi \Sigma \semi \SemEvent )
	\end{array}
	} 
\\[.6cm]
\mathrule{Cond\_true}{
	\sem{\texttt{B}}_{\ell} = {\it true}
	}{
	\C(\state \semi (\texttt{B})\{ \texttt{S} \}~\Sigma \semi \SemEvent )
	\lred{}
	\C(\state \semi \ell \semi \texttt{S}~\Sigma \semi \SemEvent )
	}
\qquad
\mathrule{Cond\_false}{
	\sem{\texttt{B}}_{\ell} = {\it false}
	}{
	\C(\state \semi (\texttt{B})\{ \texttt{S} \}~\Sigma \semi \SemEvent )
	\lred{}
	\C(\state \semi \ell \semi \Sigma \semi \SemEvent )
	}
\\
\\
\end{array}
\]
}
\caption{\label{tab:sem1}\label{tab:sem2} The transition relation of {\smallang}}
\end{table*}
%
Rule \rulename{Agree} in Table~\ref{tab:sem1} establishes the agreement of the 
parties involved in the contract. This operation is 
a multiparty synchronization, 
where some parties, namely $\vect A_i$, agree on the initial values $\vect v_i$ of the contract's fields $\wt{\texttt{{x$_i$}}}$, for $i\in 1,..n$. The resulting configuration 
moves the contract to the initial state \texttt{Q} and initializes the values of the parties parameters $\vect{\tt A}$ and the contract's fields $\vect{\tt x}$. 
We recall the syntactic conditions of the agreement term, given in the previous section: $(i)$ the sequence $\vect{\tt x}$ is a subset of the contract's fields, $(ii)$ $\bigcup_{i \in 1..n} \vect{{\tt x}_i}$ is a partition of $\vect{{\tt x}}$ and $(iii)$ the sequences of parties $\vect{{\tt A}_i}$ are subsets of the full list $\vect{\tt A}$. 
%
%

Rule \rulename{Function} defines function invocations; the label specifies the address
$A$ performing the invocation and the function name {\tt f} with the actual parameters. The 
transition may occur provided (\emph{i}) the contract is the state \texttt{Q} that admits invocations of
{\tt f} from $A$ and (\emph{ii}) the contract is \emph{idle}, 
\emph{i.e.}~the contract has no statement to execute -- \emph{c.f.}~the left-hand side runtime
contract -- (\emph{iii}) the precondition \texttt{B} is satisfied, and no event can be triggered -- \emph{c.f.}~the premise $\nored{\SemEvent}{\Time}$. In 
particular, this last constraint expresses that events  \emph{have precedence} on
possible function invocations. For example, if a payment deadline is reached and, at the same 
time, the payment arrives, it will be refused in favour of the event managing the deadline.

Rule \rulename{State Change} says that a contract changes state when the statements execution terminates and the sequence of events \texttt{W} is added to the multiset of pending events, up to the evaluations of their time expressions, \emph{i.e.} the occurrences of the 
identifier \texttt{now} are replaced by the current value of the clock. 

Rule \rulename{Event Match} specifies that event handlers may run provided there is no statement 
to perform in the runtime contract, and the time guard of the event has exactly the value 
of the global clock $\Time$. Observe that the timeouts of the events are 
evaluated in an eager way when the event is scheduled -- \emph{c.f.}~rule 
\rulename{State Change} -- not when the event handler is triggered. 
Moreover, the state change performed at the end of the execution of the event 
handler is carried over again by the rule \rulename{State Change}, 
with an empty sequence $\texttt{W}$.

Rule \rulename{Tick} models the elapsing of time. This happens when the contract has 
no statement to perform and no event can be triggered.
Intuitively, the implementation of {\smallang} on top of a blockchain will bind the global clock to the timestamp of the current block.
Therefore a sequence of semantic transitions performed in the same unit of time will correspond to a set of transactions inserted into the same block to be appended to the blockchain.

It is worth to notice that the foregoing rules imply that the complete execution of 
a function call does not affect the global time. This admits the paradoxical phenomenon 
that an endless sequence of function invocations does not make time elapse. While this is possible in theory, it is not in practice, since blocks can only include a finite number of transactions.
Additionally, all the legal contracts we have analyzed are finite state, each state admits a single function invocation, and function invocations update the state in a noncircular way,
thus preventing infinite sequence of function calls. 
An alternative choice would be to adjust the semantics so to increment the clock every 
time a maximal number of functions has been evaluated, thus forcing each block to contain at most a limited number of function invocations. We have preferred to stick to the 
simpler semantics.

%
Table~\ref{tab:sem2} defines transitions due to the execution of statements.  
All these transitions are local to the runtime contract and time does not change.
Therefore, for simplicity, we always omit the clock. 
We only discuss \rulename{Asset\_Send} and 
\rulename{Asset\_Update} because the other rules are standard.
Rule~\rulename{Asset\_Send} returns part of an asset \texttt{h} 
to the party $A$. This part, named $v$, is removed from the asset, \emph{c.f.}~the memory of
the right-hand side runtime contract in the conclusion. In a similar way, 
\rulename{Asset\_Update} moves a part $v$ of an asset \texttt{h} to an asset $\texttt{h}'$.
For this reason, the final memory becomes 
$\ell[\texttt{h} \mapsto \ell(\texttt{h})-v, \texttt{h}' \mapsto \ell(\texttt{h}') + v]$.
We observe that assets, representing physical entities (coins, houses, goods, etc.) are 
never destroyed. 
The condition $\ell(\texttt{h}) \geq v$ in the premises ensures that assets can never 
become negative. 

The semantics of {\smallang} does not consider runtime errors, for instance an attempt to drain too much value from an asset results in a stuck configuration.  We postpone to future work a precise account of runtime failures and contract errors, since it requires a deep interdisciplinary analysis of the legal issues involved in the execution of the exceptional cases.

The initial configuration of a {\smallang} program $\program$ made of a single contract $\C$ is
$
\C(\zero \semi \varnothing \semi \zero \semi \zero) \comma \Time 
$.
The contract is \emph{inactive} as long as no group of addresses has interest to run it, 
\emph{c.f.}~rule \rulename{Agree}. The global clock can be any value, because it corresponds to the absolute time, defined by the timestamp of the current block in the blockchain system.

\paragraph{Example}
Possible initial transitions of the \texttt{Bike\_Rental} contract in 
Example~\ref{ex.BikeRental} are reported in Table~\ref{tab.bikerental}. We assume that the actual names of 
parties are the same as the formal names (therefore we omit the mappings in the memories). 
Let be $\mu_0=((\mathit{Alice,Bob}), (\mathit{Alice,Bob}): (36000,2))$, \emph{i.e.}~Alice and Bob
agree about renting a bike for 2 euro for 1 hour -- the time
is measured in seconds.
Let also be $\ell = [\texttt{Lender} \mapsto \mathit{Alice},\texttt{Borrower} \mapsto \mathit{Bob},\texttt{cost} \mapsto 2, \texttt{time\_limit} \mapsto 3600]$, and
$\ell' = \ell[z \mapsto \mbox{{\tt 123}}, \mbox{{\tt use\_code}}\mapsto \mbox{{\tt 123}}]$ and \texttt{S W} be the body of the function \texttt{accept}.  

%
\begin{table*}[t]
{\small
\[
\mbox{{\tt Bike\_Rental}}(\zero \semi \emptyset \semi \zero \semi \zero) \comma \texttt{0}
\]
\[
\begin{array}{r@{\!}clr}
& \lred{\mu_0} & \mbox{{\tt Bike\_Rental}}(\texttt{Inactive} \semi \ell \semi \zero \semi \zero) \comma \texttt{0} & \mbox{\rulename{Agree}}
\\
& \lred{} &
\mbox{{\tt Bike\_Rental}}(\texttt{Inactive} \semi \ell \semi \zero \semi \zero) \comma \texttt{1}
 & \mbox{\rulename{Tick}}
\\
&\lred{\mathit{Alice}\texttt{:offer(123)}} &
\mbox{{\tt Bike\_Rental}}(\texttt{Inactive} \semi \ell[z \mapsto \mbox{\tt 123}] 
\semi z \send \mbox{{\tt use\_code}} \tostate \mbox{{\tt @Proposal}} \semi \zero) \comma \texttt{1}
& \mbox{\rulename{Function}}
\\
&\lred{} & 
\mbox{{\tt Bike\_Rental}}(\texttt{Inactive} \semi \ell' \semi \zero \tostate 
\mbox{{\tt @Proposal}} \semi \zero) \comma \texttt{1} & \mbox{\rulename{Field\_Update}}
\\
&\lred{} & 
\mbox{{\tt Bike\_Rental}}(\texttt{Proposal} \semi \ell' \semi \zero \semi \zero) \comma \texttt{1}
& \mbox{\rulename{State\_Change}}
\\
& \lred{} &
\mbox{{\tt Bike\_Rental}}(\texttt{Proposal} \semi \ell' \semi \zero \semi \zero) \comma \texttt{2}
& \mbox{\rulename{Tick}}
\\
& \lred{} & 
\mbox{{\tt Bike\_Rental}}(\texttt{Proposal} \semi \ell' \semi \zero \semi \zero) \comma \texttt{3}
& \mbox{\rulename{Tick}}
\\
&\lred{\mathit{Bob}\texttt{:accept[2]}}& 
\mbox{{\tt Bike\_Rental}}(\texttt{Proposal} \semi \ell'[y \mapsto \mbox{{\tt 2}}] \semi \texttt{S W} \tostate
\mbox{{\tt @End}} \semi \zero) \comma \texttt{3}
& \mbox{\rulename{Function}}
\end{array}
\]}
\caption{\label{tab.bikerental} Initial transitions of {\tt Bike\_Rental}}
\end{table*}

\noindent
To sum up, a legal contract behaves as follows: 
\begin{enumerate}
\item 
the first action is always an agreement, which moves the contract to an idle state;
\item 
in an idle state, fire any ready event with a matching state. If there is one, 
execute its handler until the end, which is an idle state;
\item 
if there is no event to be triggered in an idle state, \emph{either} tick \emph{or}
call any permitted function (\emph{i.e.}~with matching state and preconditions). A function
invocation amounts to execute its body until the end, which is again an idle state.
\end{enumerate}

Therefore, we observe that {\smallang} has three sources of nondeterminism: 
(\emph{i}) the order of the execution of ready event handlers, 
(\emph{ii}) the order of the calls of permitted functions, and (\emph{iii}) the delay of 
permitted function calls to a later time (thus, possibly, after other event handlers).
For example, the contract  $\C$ with two functions \texttt{@Q A:f\{\zero\}$\tostate$@Q} 
and \texttt{@Q A$'$:g\{\zero\}$\tostate$@Q} behaves as either
$\lred{A\texttt{:f}} \lred{}_n \lred{A'\texttt{:g}}$ or 
$\lred{A'\texttt{:g}} \lred{}_n \lred{A\texttt{:f}}$, where $\lred{}_n$ are transitions
that make the time elapse (rule \rulename{Tick}).
As another example, consider a contract $\C'$ with a function 
\texttt{@Q A:f \{\zero \  now $\event$\,@Q$'$\{\,"hello"$\send$A\,\}$\tostate$@Q$'$\ \}$\tostate$@Q} 
and a function \texttt{@Q A$'$:g \{\zero\}$\tostate$@Q$'$}. Then it may 
behave as either
$\lred{A\texttt{:f}} \lred{A'\texttt{:g}} \lred{\texttt{''hello''} \send \A}$ or as
$\lred{A\texttt{:f}}\lred{}_n \lred{A'\texttt{:g}}$, after which the action 
$\lred{\texttt{''hello''} \send A}$ is disabled,
or as $\lred{A'\texttt{:g}}$, which precludes the call of \texttt{f}.

\paragraph{Remark.}
The semantics of {\smallang} may be easily extended to configurations with several smart legal contracts.
It is sufficient to consider configurations as consisting of sets of runtime contracts and to
change the rule \rulename{Tick}. To illustrate the general case,
let $\contract = \C_1(\State_1\semi\ell_1\semi\Sigma_1\semi\SemEvent_1), \cdots, \C_n(\State_n\semi\ell_n\semi\Sigma_n\semi\SemEvent_n)$. 
We define 
$\Events{\contract} \eqdef \SemEvent_1 ~|~ \cdots ~|~ \SemEvent_n$. The new rule \rulename{Tick} becomes
\[
\mathrule{Tick+}{
	\nored{\Events{\contract}}{\Time}
	}{
	\contract \comma \Time \lred{} 
	\contract \comma \Time + 1
	} 
\]

\section{{\smallang} laws and Equational Theory}
\label{sec:laws}

The operational semantics of Tables~\ref{tab:sem1} and~\ref{tab:sem2} is too intensional
because it defines the behaviour of a legal contract without showing any evidence of the 
differences between contracts. Nevertheless, it is the base
for defining a more appropriate, extensional semantics using a standard technique
based on observations~\cite{MilnerBook}.
According to this technique, two contracts cannot be separated if a party using them
cannot distinguish one from the other. Said otherwise, a party 
can differentiate two contracts if he can \emph{observe different interactions}. 
It turns out that defining the observations is a critical point of the overall 
technique, because it allows to fine-tune the discriminating power of the
extensional semantics.
The appropriate observations for smart legal contracts match the design principles of {\smallang}:
let $A$ be a party, then
\begin{itemize}
\item
$A$ should observe the \emph{agreement} that he is going to sign, because stipulating a
different agreement may be unsatisfactory; therefore we envisage the observation $(A:\vect v)$, that observes the exact set of values $\vect v$ the party $A$ agreed about.
\item
$A$ should also observe the \emph{permission} or the \emph{prohibition} to 
invoke a functionality at a given time $\Time$, \emph{i.e.}~whether 
$A : \texttt{f}(\vect{v})[\vect{w}]$ is possible at $\Time$ or not;
\item
$A$ should finally observe whether at $\Time$ he can \emph{receive}
a value or an asset, \emph{i.e.}~whether $v \send A$ or $v \lolli A$ are possible 
at $\Time$ or not. 
\end{itemize}
The ordering of invocations and receives can be safely overlooked, as long as they belong 
to the same block of transactions, that is they are executed at the same global time.
Notice also that the above observations allow a party to observe contract's 
\emph{obligations}. Indeed, by shifting the observation at a specific point in time, 
one can observe the effects of executing the event that encodes a legal commitment, 
such as the issue of a sanction or the impossibility to do further actions.
On the other hand, the foregoing notion of observations abstracts away from 
the names of the contract's assets and internal states.

We will use the following notations:
\begin{itemize}
\item Let be $\alpha_1= (\vect A,\, {\vect A_1}{:} \vect v_1 \cdots {\vect A_n}{:}\vect v_n)$ and $\alpha_2= (\vect B,\, {\vect B_1}{:} \vect w_1 \cdots {\vect B_n}{:}\vect w_n)$ then we write $\alpha_1 \sim \alpha_2$ if $\vect A$ and $\vect B$ are equal up to reordering of sequences, and similarly for  ${\vect A_i}: \vect v_i$ and $\vect B_j:\vect v_i$.

\item
Let be $\Lred{\mu} \; \eqdef \; \Lred{}\lred{\mu} \Lred{}$, where $\Lred{}$ stands for any number of $\lred{}$ transitions, possibly zero.
\end{itemize}

%
The following definition of legal contract equivalence compares the observable behavior of contracts. It is defined over configurations, so to appropriately shift the time of the contract's observations. The equivalence is defined as a suitable bisimulation game
that is consistent with the idea that in blockchain systems the interactions are batched in blocks of transactions.

\begin{definition}[Legal Bisimulation]
\label{def.simulation}
A symmetric relation $\mathcal{R}$ is a 
legal bisimulation between two configurations at time $\Time$,
written $\contract_1 {\comma} \Time \ \ \mathcal{R}\ \ \contract_2 {\comma} \Time$, whenever
\begin{enumerate}
\item
if $\contract_1 {\comma} \Time \Lred{\alpha} \contract_1' {\comma} \Time$ then 
$\contract_2 {\comma} \Time \Lred{\alpha'} \contract_2' {\comma} \Time$ for some
$\alpha'$ such that
  $\alpha\sim\alpha'$ and 
$\contract_1' {\comma} \Time \ \ \mathcal{R} \ \ \contract_2', \Time$;
\item
if $\contract_1 {\comma} \Time \Lred{\mu_1} \cdots \Lred{\mu_n} \contract_1', \Time \lred{}\contract_1', \Time+1$
then there exist $\mu_1' \cdots \mu_n'$ that is a permutation of $\mu_1 \cdots \mu_n$ 
such that
$\contract_2 {\comma} \Time \Lred{\mu_1'} \cdots \Lred{\mu_n'} \contract_2'{\comma}\Time
 \lred{} \contract_2'{\comma} \Time+1$  and $\contract_1' {\comma} \Time+1 \ \mathcal{R} \ \contract_2' {\comma} \Time+1$.
\end{enumerate}
Let $\simeq$ be the largest legal bisimulation, called \emph{bisimilarity}.
When the initial configurations of contracts $\C$ and $\C'$ are bisimilar, 
we simply write $\C \simeq \C'$.
\end{definition}

Being a symmetric relation, a legal bisimulation compares both 
contracts' permissions and prohibitions: if $\C$ permits an action 
(\emph{i.e.}~exhibits an observation), then $\C'$ must permit the same action, 
and if $\C$ prohibits an action (\emph{i.e.} does not exhibit a function call or an external 
communication), then also $\C'$ must not exhibit the corresponding observation. 
Moreover, the bisimulation game enforces a \emph{transfer property}, that is it
shifts the time of observation to the future, 
so to capture and compare the changes of permissions/prohibitions and the (future) 
obligations. 
Observe that the equivalence abstracts away the ordering of the observations within the same time clock, since in a blockchain there is no strong notion of ordering between the transactions contained in the same block. 
Nevertheless, specific orderings of function invocations are important in {\smallang}
contracts and the equivalence cannot overlook essential precedence constraints.
For instance, the requirement that a function delivering a service 
can only be invoked after another specific function, say a payment. 
This is indeed the case for the legal bisimulation. 
To explain, consider the contract  $\C$ with two functions 
\texttt{@Q A:f\{\zero\}$\tostate$@Q} and \texttt{@Q A$'$:g\{\zero\}$\tostate$@Q} 
and the contract $\C'$ 
with two functions \texttt{@Q A:f\{\zero\}$\tostate$@Q$'$} 
and \texttt{@Q$'$ A$'$:g\{\zero\}$\tostate$@Q}. In $\C$ the functions can be called in any order, while in $\C'$ the function \texttt{g} can be invoked only after \texttt{f}.
Accordingly, $\C\not\simeq\C'$ since there is a runtime configuration of $\C$ that 
exhibits the observation $\Lred{A' : {\rm g}}$, while it is not the case for 
the contract $\C'$, since at any time it can only exhibit 
$\Lred{A : {\rm f}}\Lred{A' : {\rm g}}$.

The following theorem highlights the property that the internal state of the contract is abstracted away by the extensional semantics, which only observes the external contract's behavior.
The proof is omitted because it is standard.

\begin{theorem}[Internal refactoring]\label{thm:rename}
Let $\C$ and $\C'$ be two contracts that are equal up-to a bijective renaming of states.
Then $\C \simeq \C'$. Similarly, for bijective renaming of assets, fields and contract names.
\end{theorem}

The theorem could be extended to contracts equal up-to the number of contract's fields and assets, as long as their external behavior is the same. We keep for future work a precise formalization of the internal refactoring allowed by the observational equivalence.

Bisimilarity is also independent from \emph{future} clock values. This 
allows us to garbage-collect events that cannot be triggered 
anymore because the time for their scheduling is already elapsed. 

\begin{theorem}[Time shift]\label{thm:time}
\mbox{ }
\begin{enumerate}
\item 
If \, $\contract {\comma} \Time\simeq\contract' {\comma} \Time$ 
and $\Time \leq \Time'$, then $\contract {\comma}\Time'\simeq\contract' {\comma} \Time'$.
\item 
If \, $\Time < \Time'$ then $\C(\state \semi \ell \semi \Sigma \semi \SemEvent~|~\Time \event \texttt{@}\state~\{~\texttt{S}~\} \tostate \texttt{@}\state') \,{\comma}\, \Time' \ \simeq\ \C(\state \semi \ell \semi \Sigma \semi \SemEvent ) \,{\comma}\, \Time'$.
\end{enumerate}
\end{theorem}

We finally put forward a set of algebraic laws that formalize the fact that the 
ordering of remote communications can be safely overlooked, as long as they belong 
to the same transaction.
The laws are defined over statements, therefore, 
let $\C[~]$ be a \emph{context}, that is a contract that contains an hole where a statement may occur. We write $\S \simeq \S'$ if, for every context 
$\C[~]$, $\C[\S] \simeq \C[\S']$.

\begin{theorem}
\label{thm.noninterferencelaws}
The following non-interference laws hold in {\smallang}
(whenever they are applicable, we assume $\texttt{x} \notin \fv{\texttt{E}'}$
and $\texttt{x}' \notin \fv{\texttt{E}}$ and $\texttt{h} \notin \fv{\texttt{E}'}$
and $\texttt{h}' \notin \fv{\texttt{E}'}$ and $\texttt{h}'' \notin \fv{\texttt{E}}$
and $\texttt{h}''' \notin \fv{\texttt{E}}$):
\[
\begin{array}{@{\hspace{-.0cm}}r@{\quad}c@{\quad}l@{\quad}l}
\texttt{E} \send \A \; \; \texttt{E}' \send \A' & \simeq & \texttt{E}' \send \A' \; \; \texttt{E} \send \A
& 
\\
\texttt{E} \send \texttt{x} \; \; \texttt{E}' \send \A & \simeq & \texttt{E}' \send \A \; \; 
\texttt{E} \send x
&  
\\
\texttt{E} \send \texttt{x} \; \; \texttt{E}' \send x' & \simeq & \texttt{E}' \send \texttt{x}' \; \; \texttt{E} \send \texttt{x}
&  
\\
\texttt{E} \lolli \texttt{h},\A \; \; \texttt{E}' \send \A' & \simeq & \texttt{E}' \send \A' \; \; \texttt{E} \lolli \texttt{h},\A
&  
\\
\texttt{E} \lolli \texttt{h},\A \; \; \texttt{E}' \send \texttt{x}' & \simeq & \texttt{E}' \send \texttt{x}' \; \; \texttt{E} \lolli \texttt{h},\A
&  
\\
\texttt{E} \lolli \texttt{h},\texttt{h}' \; \; \texttt{E}' \send \A & \simeq & \texttt{E}' \send \A \; \; \texttt{E} \lolli \texttt{h},\texttt{h}'
& 
\\
\texttt{E} \lolli \texttt{h},\texttt{h}' \; \; \texttt{E}' \send \texttt{x}' & \simeq & \texttt{E}' \send \texttt{x}' \; \; \texttt{E} \lolli \texttt{h},\texttt{h}'
& 
\\
\texttt{E} \lolli \texttt{h},\A \; \; \texttt{E}' \lolli \texttt{h}'',\A' & \simeq & \texttt{E}' \lolli \texttt{h}'',\A' \; \; \texttt{E} \lolli \texttt{h},\A
&  
\\
\texttt{E} \lolli \texttt{h},\A \; \; \texttt{E}' \lolli \texttt{h}'',\texttt{h}''' & \simeq & \texttt{E}' \lolli \texttt{h}'',\texttt{h}''' \; \; \texttt{E} \lolli \texttt{h},\A
&  
\\
\texttt{E} \lolli \texttt{h},\texttt{h}' \; \; \texttt{E}' \lolli \texttt{h}'',\texttt{h}''' & \simeq & \texttt{E}' \lolli \texttt{h}'',\texttt{h}''' \; \; \texttt{E} \lolli \texttt{h},\texttt{h}'
&  
\end{array}
\]
\end{theorem}

\begin{proof}
We prove the first equality. 
Let be $\S_1=\texttt{E} \send \A \; \; \texttt{E}' \send \A'$
and $\S_2=\texttt{E}' \send \A' \; \; \texttt{E} \send \A$, and let be $\C_1=\C[\S_1]$ and $\C_2=\C[\S_2]$,  then we need to prove that
$\C_1(\zero,\varnothing,\zero,\zero){\comma}\Time \simeq \C_2(\zero,\varnothing,\zero,
\zero){\comma}\Time$. 
Let also 
$\contract_i=\C_i(\State,\ell,\zero,\Psi[\S_i])$, with $i=1,2$, be the runtime contract 
where the statement $\S_i$ occurs within a number of handlers of future events.
 We demonstrate
that the symmetric closure of the following relation is a legal bisimulation:
$$
\begin{array}{l}
\{\, \bigl(\C_1(\zero,\varnothing,\zero,\zero){\comma}\Time \ ,\ \C_2(\zero,\varnothing,\zero,\zero){\comma}\Time \bigr) \, \} 
\\[2mm]
\cup 
\, \{\, \bigl(\C_1(\texttt{Q},\ell,\zero,\Psi[\S_1]){\comma}\Time' \ ,\ \C_2(\texttt{Q},\ell,\zero,\Psi
[\S_2]){\comma}\Time'\bigr) 
\quad \; |~ \; \text{for every } \texttt{Q}, \ell, \Psi[~], \Time' \; \}
\end{array}
$$
Indeed, notice that the statement $\texttt{E} \send \A \; \; \texttt{E}' \send \A'$ can 
only contribute to the behavior of $\C$ with a couple of transitions during the evaluation of the body of a function or the evaluation of an event handler. Therefore the statement must be completely executed with the same time clock, possibly a number $k$ of times due to the multiple function calls and event handlers that are executed during the same time clock.

Formally, if $\contract_1,\Time' \Lred{\mu_1}\cdots\Lred{\mu_n}\contract_1',\Time'\lred{}\contract_1',\Time'+1$, then the sequence $\mu_1\cdots\mu_n$ contains $k$ occurrences of the pair $v\send\A, v'\send\A'$. Similarly, there exist $\mu'_1\cdots\mu'_n$ and a configuration 
$\contract_2',\Time'$ such that
$\contract_2,\Time' \Lred{\mu'_1}\cdots\Lred{\mu'_n}\contract_2',\Time'
\lred{}\contract_2',\Time'+1$, where the sequence $\mu'_1\cdots\mu'_n$ 
is identical to $\mu_1\cdots\mu_n$ but for the $k$ occurrences of the pair 
$v\send\A, v'\send\A'$ that has been swapped into $v'\send\A', v\send\A$. The argument 
also holds in the converse direction.
\end{proof}


\section{Towards a distributed implementation}
\label{sec:expressivepower}

The definition of {\Stipula} is implementation-agnostic, that is it can be either
executed as a centralized application or it can be run on top of a distributed system, 
such as a blockchain. 
Implementing {\Stipula} in terms of smart contracts, \emph{e.g.}~{\Solidity}, 
would bring in the advantages of a public and decentralized blockchain platform. These include the 
benefits coming form the fact that several governments have recently recognised  that  smart contracts and, more generally, programs operating over distributed ledgers, have legal 
value~\cite{WyomingAct,ItalianLaw2019,MaltaAct}. 
However, {\Stipula}'s contracts
are more general and encompass smart contracts: they provide benefits in terms of 
automatic execution and enforcement of contractual conditions, 
traceability, and outcome certainty even without using a blockchain. 
In particular, running a legal contract over a secured centralized system 
allows for more efficiency, energy save, additional privacy. Moreover, 
a controlled level of intermediation can better monitor the contract enforcement, 
dealing with disputes between contract's parties and carring out judicial 
enforcements.
Finally, the intrinsic open nature of legal contracts is another challenge 
for smart contracts, that can hardly deal with the off-chain world: external data 
enter the blockchain only through oracles, which are problematic in many senses, 
and the dynamic change of behaviours conflicts with the rigidity of smart contracts 
definition. Time is another big issue in blockchains.

%
While we are currently prototyping {\Stipula} on a centralized system, we think that 
a distributed implementation on a blockchain system is interesting. 
%
%
Below we discuss the issues of prototyping {\Stipula}, 
either on top of a centralized system or a distributed system like a (public) blockchain. 

\paragraph{Functions and States.} {\Stipula} contracts are 
very close to class definitions, therefore sticking to an object-oriented target language
as Java or Solidity 
would ease the prototyping effort. In this context, the state-aware programming 
is also well developed. In particular, several smart contract languages widely use
the state machine pattern (\emph{c.f.}~{\Solidity}~\cite{SoliditySM} and {\sf Obsidian}~\cite{Obsidian}). 

%


\paragraph{Parties.}
The implementation of {\smallang} must carefully handle digital identities, ensuring that a contract's function is actually invoked by the correct caller. In blockchain implementation this corresponds to the \emph{externally owned accounts}, but the pseudoanonimity provided by \emph{public} blockchains might be a limitation, since legal contract's parties must have a trusted identity, especially in the case of authorities, which must be bound to parties that are trusted intermediaries.

\paragraph{Agreements.} The agreement code  
is technically a \emph{join synchronization} that expresses a consensus between the 
parties to start the contract with particular values of the (non-linear) fields. 
This construct can be implemented by resorting to a barrier-like protocol, 
where each party \texttt{A$_i$} may call, in whatever order, a specific function 
to propose the values he agrees on, and the barrier eventually checks the consistency 
of the proposed values before moving the contract to the initial state. 
The following snippet of {\Solidity}-like code 
(a similar code can be written also in Java RMI)
corresponds to the agreement code (see Section~\ref{sec:smartlegalcontracts}), where we 
assume that the fields of the contract are 
{\tt x$_1$}, $\cdots$, {\tt x$_k$} with types {\tt T$_1$},$\cdots$,{\tt T$_k$}, respectively. 
We also assume that $\vect{{\tt A}_i} : \vect{{\tt x}_i}$ is defined as
${{\tt A}_i}^1, \cdots , {{\tt A}_i}^{s\_i} : {{\tt x}_i}^1, \cdots , {{\tt x}_i}^{s\_i'}$.

{\footnotesize
\begin{lstlisting}[mathescape,basicstyle=\ttfamily]
  address A$_1$, ... , A$_n$ ;
  T$_1$ x$_1$; ... ; T$_k$ x$_k$ ;
  enum State {Nothing, Q} 
  State state = State.Nothing ;        
  int counter = 1 ;        

  T${_i}^1$ aux$_i\_$x${_i}^1$;...; T${_i}^{r\_i'}$ aux$_i\_$x${_i}^{r\_i'}$;                                             // 1 $\leq i\leq$ n
  bool use$\_$once$_i$ = true ;                                           // 1 $\leq i\leq$ n
 
  function set$\_$ok$_i$(T${_i}^1$ z${_i}^1$,...,T${_i}^{r\_i}$ z${_i}^{r\_i}$,T${_i}^{r\_i+1}$ z${_i}^{r\_i+1}$,...,T${_i}^{r\_i+r\_i'}$ z${_i}^{r\_i+r\_i'}$){		// 1 $\leq i\leq$ n
      if (sender == A$_i$ $\&\&$ use$\_$once$_i$){
             use$\_$once$_i$ = false ;
             x${_i}^1$ = z${_i}^1$ ; ... ; x${_i}^{r\_i}$ = z${_i}^{r\_i}$ ; 
             aux$_i\_$x$_1$ = z${_i}^{r\_i + 1}$;...; aux$_i\_$x$_{r_i'}$ = z${_i}^{r\_i + r\_i'}$;
             if (counter == n){
                    if ($\bigwedge_{1 \leq \texttt{i} \leq n}$(aux$_i\_$x${_i}^1$==x${_i}^1$ $\&\&$...$\&\&$ aux$_i\_$x${_i}^{r\_i'}$==x${_i}^{r\_i'}$)) 
                               state = State.Q ;
                    else throw error;
             } else counter = counter + 1 ;
      } else throw error; 
}
\end{lstlisting}}

Each function \texttt{set\_ok$_i$} can be called only once by the party {\tt A$_i$}, 
with two lists of parameters: the first $r_i$ values are used to set the contract's 
fields {\tt x${_i}^1$}, $\cdots$, {\tt x${_i}^{r\_i}$}. 
The last $r_i'$ values are recorded into the auxiliary fields 
{\tt aux$_i\_$x${_i}^{r\_i + 1}$}, $\cdots$, {\tt aux$_i\_$x${_i}^{r\_i +r\_i'}$}, to express that 
{\tt A$_i$} agrees with anyone setting the contract's fields 
{\tt x${_i}^{r\_i + 1}$}, $\cdots$, {\tt x${_i}^{r\_i +r\_i'}$}.
When all the parties have done the agreement, \emph{i.e.}~\texttt{counter} 
is equal to \texttt{n}, a check on the consistency of {\tt aux$_i\_$x${_i}^j$} is performed
and, in case it succeeds, the contract becomes active moving to the state {\tt Q}.
The snippet also shows that contract's states can be easily mapped to enumerations, 
as usual in the {\Solidity} State Machine pattern.

There is a discrepancy between the above code and the semantics of the agreement 
in Table~\ref{tab:sem1}. While the agreement has a transactional nature (it may occur as 
a whole or not), the above {\Solidity} protocol takes time, \emph{i.e.}~it is performed in
several blocks and, in any block, a failure may occur. In this case, there is a
backtrack to the \emph{initial state of the block} and not to the initial state of
the protocol, as it happens in {\smallang}. This means that the error management should
take care of removing partial values stored in the fields of the contract. Nevertheless,
another source of discrepancy seems more awkward: the gas consumption. In fact, the 
successful termination of the agreement as well as the failed one have a cost in 
Ethereum, while it is not the case. To bridge this gap, we should design agreements with 
fees payed by parties that are used for the consensus. We have not yet studied these
details, which are postponed to future work.



\paragraph{Assets.}
Assets are linear resources that cannot be duplicated or leaked:  
when a resource value is assigned, the location previously holding the value is emptied. 
Modern programming languages, $e.g.,$ {\sf Rust} and {\sf Move}, have already recognized the relevance of having 
linear resources as first-class entities, 
because they can significantly simplify programming and the effort required for verification.
{\smallang} features a simple abstraction to manage assets, which is used to 
represent both currency and indivisible tokens. To implement these assets 
on top of a blockchain, we can
resort to the popular token standards on Ethereum (ERC-20 for virtual currency 
and ERC-721 for non fungible tokens\cite{ERCTokens}). Alternatively, we can rely on the 
{\sf Move} language, whose designers have featured programmable linear resources 
by constraining them to adhere to ad-hoc rules specified by its declaring 
module~\cite{Resources20,Move}. Using a pseudo-code inspired to {\sf Move}, we might define (divisible) assets \texttt{h} and \texttt{h$'$} as  resources of type 
\texttt{H}, so that the operations $\texttt{E} \lolli \texttt{h},\texttt{h}'$ 
and $\texttt{E} \lolli \texttt{h},\texttt{A}$  may be encoded by 
\texttt{h.move(E,h$'$)} and \texttt{h.withdraw(E,A)}, 
according to the definitions below:

{\footnotesize
\begin{verbatim}
 resource H { 
     T amount ;
		
     function move(T x, H h) {
          (x <= amount){ h = h + x ; amount = amount - x ; }
     }		
     function withdraw(T x, address A) {
          (x <= amount){ A.send(x) ; amount = amount - x ; }
     }
     constructor(T x){ amount = x ; }
 }
\end{verbatim}
}


\paragraph{Events.}
Events correspond to scheduling a computation for future execution. 
While this is a common feature of concurrent programming
(many mainstream languages provide primitives for \emph{futures} and \emph{callbacks}), 
it is more difficult to implement events in the context of blockchain.
There are two reasons:
\begin{enumerate} 
\item
in blockchains, the time flows according to the block insertion. Therefore, if the event 
should be scheduled in one minute and the next block is inserted in ten minutes, there is
a delay that must be pondered;
\item
blockchains do not admit the record of statements that 
have to be performed afterwards in a future transaction. 
\end{enumerate}
As regards the second issue, the standard technique adopted in Ethereum to 
circumvent this 
limitation
is based on the {\Solidity}'s \emph{event} construct and off-chain \emph{oracle} services.
To explain, the {\smallang} event 
\texttt{E$\, \event \,$ @Q \{ S \}$\, \tostate \,$ @Q$'$} emitted by a function {\tt foo} 
can be mapped to the
{\Solidity} code:

{\footnotesize
\begin{lstlisting}[mathescape,basicstyle=\ttfamily]
    event R(uint time, address lc) ;
    function call_back_R() external {
        if (state == State.Q) { S ; state = State.Q'; }    
    }
    function foo(T$_1$ u$_1$,..., T$_n$ u$_n$){
        ...
        emit R(E,address(this)) ;
    }
\end{lstlisting}}
\noindent
The {\Solidity} function \texttt{foo} emits an event named \texttt{R} carrying 
the time \texttt{E}
and the address of the issuing contract. 
Moreover, an external \emph{DApp service} -- an oracle -- scans the blockchain 
looking for \texttt{R} events and calls the \texttt{call\_back\_R} function of 
the contract at the appropriate time \texttt{E}.
Other more complex and safer protocols can be used. However a code external
to the Ethereum blockchain is always necessary in order to trigger the scheduled 
event handler. 
A more satisfactory implementation of the {\smallang} events might 
adopt ideas taken from the implementation of timeouts in the {\sf Marlowe} on top
of the Cardano blockchain,
as discussed in Section~\ref{sec:relatedworks}.   

To conclude the section, notice that, in this paper, we are overlooking issues regarding errors and backtracks.

\section{Expressivity of {\smallang}}
\label{sec:examples}

{\smallang} has been devised 
for writing legal contracts in a formal and intelligible (to lawyers) way. In this
section we analyze the expressivity of {\smallang} by 
writing the contracts for a set of archetypal acts, ranging from renting to (digital) 
licenses and to bets. 
We conclude this analysis with a table that summarizes the legal elements of the archetypal 
acts and the programming abstractions that we have used to express them in {\smallang}.
%

\subsection{The free rent contract}

 {\footnotesize
 \begin{figure}[t]
\begin{lstlisting}[numbers=left,numberstyle=\tiny,mathescape,basicstyle=\ttfamily,caption={The rent for free contract},captionpos=b,label=freerent] 
stipula Free_Rent {  

  assets token 
  fields numBox, t_start, t_limit

  agreement (Lender, Borrower)(t_start, t_limit) {
       Lender , Borrower : t_start, t_limit      
  } $\tostate$ @Inactive

  @Inactive Lender :  boxProposal (n)[t] {
       t $\lolli$ token 
       n $\send$ numBox 
       now + t_start $\event$ @Proposal { 
             token $\lolli$ Lender } $\tostate$ @End
  } $\tostate$ @Proposal

  @Proposal Borrower :  boxUse {
       (uses(token), numBox) $\send$ Borrower 
       now + t_limit  $\event$  @Using { 
             "Time_Limit_Reached" $\send$ Borrower 
             token $\lolli$  Lender 
       } $\tostate$ @End
  } $\tostate$ @Using
   
  @Using Borrower : returnBox {
       token $\lolli$ Lender 
  } $\tostate$ @End
} 
\end{lstlisting}
\end{figure}
}

The free rent is the simplest kind of legal contract. It involves two parties, 
the lender and the borrower, which initially agree about what good is rented, 
what use should be made of it, the time limit (or in which case it must be returned), 
the estimated of value and any defects in the good. 
Upon agreement, the delivery of the good triggers the legal bond, that is the 
borrower has the permission to use the good and the lender has the prohibition 
of preventing him from doing so. Note that there is no transfer of ownership, 
but only the right to use the good. 
The contract terminates either when the borrower returns the good, or when the 
time limit is reached. Litigations could arise when the borrower violates the 
obligations of diligent storage and care, the obligations of using the good only as intended, and not granting the use to a third party without the lender's consent. In these cases the lender may demand the immediate return of the object, in addition to compensation for the damage. On the other hand, the borrower is entitled to compensation if the good has defects that were known to the lender but that he did not initially disclose.


\medskip\noindent
The free rent contract puts forward the following points:
\begin{itemize}
\item When a legal contract refers to a \emph{physical} good, the smart contract needs a digital handle (an avatar) for that good. Many technological solutions, such as smart locks of IoT devices, are actually available. In {\smallang} we abstract from the specific nature of such a digital handle, and we simply represent it as an \texttt{asset}, which intuitively corresponds
to a non fungible token associated to the physical good.
\item The rent legal contract grants just the \emph{usage} of a good without the transfer of ownership. Accordingly, while the communication of the token provides full control of the associated physical good, we assume an operation \texttt{uses(token)} (resp. \texttt{use\_once(token)} or \texttt{uses(token,A)}) that generates a usage-code providing access to the object associated to the token (resp. a usage-code only valid (once) for the party \texttt{A}).
\item 
In a legal rent contract it is important to acknowledge the delivery of the good, since this is the action that triggers the legal bonds. We rely on assets and their \emph{semantics} 
to implement this feature.
\end{itemize}

The {\smallang} code for the free rent of a locker 
is written in Listing~\ref{freerent}.
The two parties agrees on the time limit 
for the locker usage (\texttt{time\_limit}) and the time limit to start the 
usage (\texttt{time\_start}). 
Contract's states allow one sequence of actions: first {\tt Lender} 
sends the number {\tt n} of the locker and the token {\tt t} associated to {\tt n} by calling \texttt{box\_proposal} (line 10). 
This action moves the contract to the temporary state \texttt{Proposal} and 
schedule the event in line 13. This event is essential to prevent the 
unique token associated to the locker to be indefinitely locked-in in the 
smart contract when the borrower never calls the \texttt{boxUse} function 
to finalize the delivery of the good.
If {\tt Borrower} calls the function \texttt{boxUse} 
(line 16) within the timeout \texttt{time\_start}, then the number of the rented locker 
and  the access code are returned. 
At the same time, a second timeout is installed to check the time limit for the 
locker usage, and the final state change to \texttt{Using} (line 21) vanishes
the timeout installed in line 13. 
This second event is needed to prevent a
never ending use of the locker.
If \texttt{time\_limit} is reached and the contract's 
state is still \texttt{Using}, then (lines 18-19) a message is sent to {\tt Borrower}
 and the token is sent back to {\tt Lender}, which becomes again in full control of
  the locker and can thus invalidate the access code held by the borrower. 
Otherwise, the rent contract can terminate
because the borrower explicitly returns the good before the time limit. 
This is represented by a call of the function \texttt{returnBox} (line 23).

\noindent
The smart legal contract of Listing~\ref{freerent} does not consider the compensations 
that would be needed to deal with the disputes due to breaches of the contract, such 
as the borrower's diligent care during the locker's usage. 
These violations require off-chain monitoring and a dispute resolution mechanism. 
The next example illustrates how this off-chain monitoring and enforcement can be 
combined with the on-chain code.

\subsection{The Digital Licensee contract}

{\footnotesize
 \begin{lstfloat}
\begin{lstlisting}[numbers=left,numberstyle=\tiny,mathescape,basicstyle=\ttfamily,caption={The contract for a digital licence},captionpos=b,label=DigitalLicence] 
stipula Licence {
  assets token, balance  
  fields cost, t_start, t_limit 
 
  agreement (Licensor,Licensee,Authority)(cost, t_start, t_limit) {  
    Licensor, Licensee : cost, t_start, t_limit 
   } $\tostate$ @Inactive

  @Inactive Licensor : offerLicence [t] { 
    t $\lolli$ token 
    now + t_start $\event$ @Proposal { 
               token $\lolli$ Licensor } $\tostate$ @End
  } $\tostate$ @Proposal

  @Proposal Licensee : activateLicence [b]
    (b == cost){
       b $\lolli$ balance
       balance*0,1 $\lolli$  balance, Authority 
       uses(token,Licensee) $\send$ Licensee 
       now + t_limit $\event$  @Trial {
                      balance $\lolli$ Licensee 
                      token $\lolli$  Licensor 
                    } $\tostate$ @End
  } $\tostate$ @Trial

  @Trial Licensee : buy {
    balance $\lolli$ Licensor 
    token $\lolli$ Licensee 
  } $\tostate$ @End

  @Trial Authority : compensateLicensor {
    balance $\lolli$  Licensor 
    token $\lolli$  Licensor 
  } $\tostate$@End

  @Trial Authority : compensateLicensee {
    balance $\lolli$ Licensee 
    token $\lolli$ Licensor; 
  } $\tostate$ @End
}
\end{lstlisting}
\end{lstfloat}
}

Let us consider a contract corresponding to a licence to access a digital service, like a software or an ebook: 
the digital service can be freely accessed for a while, and can be permanently bought with an explicit communication within the evaluation period (for a similar example, see \cite{Sartor2018}). The licensing contractual clauses  can be described as follows:
\begin{description}
\item[Article 1.]
{\tt Licensor} grants {\tt Licensee} for a licence to evaluate the {\tt Product}
 and fixes (\emph{i}) the
\emph{evaluation period} and (\emph{ii}) the \emph{cost} of the {\tt Product} if 
{\tt Licensee} will 
bought it.
\item[Article 2.] 
{\tt Licensee} will pay the {\tt Product} in advance; he will be reimbursed if the 
{\tt Product} 
will not be bought with an explicit communication within the evaluation period.
 The refund will be the 90\% of the cost
because the 10\% is payed to the {\tt Authority} (see Article 3).
\item[Article 3.] 
{\tt Licensee} must not publish the results of the evaluation during the evaluation period
and {\tt Licensor} must reply within 10 hours to the queries of {\tt Licensee} related to the 
{\tt Product}; this is supervised by {\tt Authority} that may interrupt the licence and reimburse
either {\tt Licensor} or {\tt Licensee} 
according to whom breaches this agreement.
\item[Article 4.]
This license will terminate automatically at the end of the evaluation period, if the licensee does not buy the product.
\end{description}

Compared to the previous example, this contract involves payment and refund: 
an amount of currency is escrowed, and two parts of it will be sent to different 
parties, the {\tt Authority} and either the {\tt Licensor} or 
the {\tt Licensee}. {\smallang} provides the general \texttt{asset} abstraction, 
together with a general operation to move just a (positive) subset of the asset
 to a different owner. 
This is exactly what is needed to deal with currency, therefore the {\smallang} licence contract holds two different assets: an indivisible \texttt{token} providing an handle to the digital service, and a \texttt{balance} that is a divisible asset corresponding to the amount of currency kept in custody inside the smart contract.

A further important feature of the contract is Article 3 that defines specific constraints 
about the \emph{off-chain} behaviour of {\tt Licensor} and {\tt Licensee}. 
This exemplifies the very general situations where contract's violations cannot 
be fully monitored by the on-chain software, such as the publication of a post in 
a social network, or the leakage of a secret password, or the violation of the 
obligation of diligent storage and care. In all these cases, it is required a 
trusted third party, say an {\tt Authority}, to supervise the disputes occurring 
from the off-chain monitoring and to provide a trusted dispute resolution mechanism.
The code in Listing~\ref{DigitalLicence} illustrates the encoding
of the off-chain monitoring and enforcement mechanism with the on-chain smart 
contract code in {\smallang}.



The \texttt{agreement} of Listing~\ref{DigitalLicence} involves three parties: 
\texttt{Licensor}, which fixes the parameters of the contract, according to Article 1., 
\texttt{Licensee}, which explicitly agrees, and \texttt{Authority}, which does not 
need to agree upon the contracts' parameters (\emph{i.e.}~the emptyset agreement noted 
$\zero$), but it is important that it is involved 
in the agreement synchronization, because it plays the role of the trusted third party 
that is entitled to call the functions \texttt{compen}-\texttt{sateLicensor} and 
\texttt{compensateLicensee}.
%

In \texttt{activateLicence}, the caller, \emph{i.e.}~the 
{\tt Licensee}, is required to send an amount of assets equal to the fixed cost of 
the license. 
Notice the difference between line 18 and line 19: the first one is the move 
of a fraction of asset towards the authority, while the second is the 
simple communication to {\tt Licensee} of a personal usage code associated to the token. 
Once entered in the \texttt{Trial} state, the contract can terminate in three ways: 
(\emph{i}) the licensee expresses its willingness to buy the licence by calling 
the function \texttt{call} which grants him the full token, or (\emph{ii})
 the time limit for the free evaluation period is reached, thus the scheduled event 
refunds the licensee and gives the token back to the licensor, or 
(\emph{iii}) during the evaluation period a violation to Article 3 is identified and the authority pre-empts the license by calling either the function \texttt{compensateLicensor} or \texttt{compensateLicensee}. Observe that it is important that the code guarantees that, in all the possible cases, the assets, both the token and the balance, are not indefinitely locked-in the contract.

\subsection{A bet contract}

The bet contract is a simple example of a legal contract that contains an element 
of randomness (\emph{alea}), \emph{i.e.}~where the existence of the performances 
or their extent depends on an event which is entirely independent of the will of 
the parties.
The main element of the contract is a future, aleatory event, such as the 
winner of a football match, the delay of a flight, the future value of a company's stock.

A digital encoding of a bet contract requires that the parties explicitly 
agree on the source of data that will determine the final value of the 
aleatory event (the \emph{Data Provider}), that is a specific online service, 
an accredited institution, or any trusted third party. 
It is also important that the digital contract provides precise time limits 
for accepting payments and for providing the actual value of the aleatory event. 
Indeed there can be a number of issues: the aleatory event does not happen, 
\emph{e.g.}~the football match is cancelled, or the data provider fails to provide 
the required value, \emph{e.g.}~the online service is down.

{\footnotesize
\begin{lstlisting}[numbers=left,,numberstyle=\tiny,mathescape,basicstyle=\ttfamily,,caption={The contract for a bet},captionpos=b,label=BetContract] 
stipula Alea {
  assets bet1, bet2
  fields alea_fact, val1, val2, data_source, 
         fee, amount, t_before, t_after

  agreement(Better1,Better2,DataProvider)
            (fee, data_source, t_before, t_after, alea_fact, amount){
 
    DataProvider, Better1, Better2 : fee, data_source, t_after, alea_fact
    Better1, Better2 : amount, t_before                  
   } $\tostate$ @Init
    
  @Init Better1 : place_bet(x)[y]
    (y == amount){ 
        y $\lolli$ bet1
        x $\send$ val1
        t_before $\event$ @First { bet1 $\lolli$ Better1 } $\tostate$ @Fail
    } $\tostate$ @First

  @First Better2: place_bet(x)[y]
    (y == amount){ 
        y $\lolli$ bet2
        x $\send$ val2  
        t_after $\event$ @Run {
            bet1 $\lolli$ Better1
            bet2 $\lolli$ Better2 } $\tostate$ @Fail
  } $\tostate$ @Run
    
  @Run DataProvider : data(x,z)[]
    (x==alea_fact){ 
        (z==val1){              // The winner is Better1             
            fee $\lolli$ bet2,DataProvider
            bet2 $\lolli$ Better1
            bet1 $\lolli$ Better1
        }     
        (z==val2){              // The winner is Better2             
            fee $\lolli$ bet1,DataProvider
            bet1 $\lolli$ Better2
            bet2 $\lolli$ Better2
        }
        (z != val1 $\&\&$ z != val2){          //No winner
            bet1 $\lolli$ DataProvider
            bet2 $\lolli$ DataProvider
        } 
   } $\tostate$ @End   
} 
\end{lstlisting}
}

The {\smallang} code in Listing~\ref{BetContract} corresponds to the case where \texttt{Better1}, respectively
\texttt{Better2}, places \texttt{val1}, respectively \texttt{val2}, a bet corresponding 
to the agreed \texttt{amount} of currency, stored in the contract's assets \texttt{bet1} 
and \texttt{bet2} respectively\footnote{For simplicity, this code requires {\tt Better1} 
to place its bet before {\tt Better2}, however it is easy to add similar function to let
the two bets be placed in any order.}. Observe that both bets must be placed within 
an (agreed) time limit \texttt{t\_before} (line 17), to ensure that the legal bond 
is established before the occurrence of the aleatory event. 
The second timeout, scheduled in line 24, is used to ensure the contract termination 
even if the \texttt{DataProvider} fails to provide the expected data, through the 
call of the function \texttt{data}. 

Compared to the Digital Licence in Listing~\ref{DigitalLicence}, the role of the 
\texttt{DataProvider} here is less pivotal than that of the \texttt{Authority}. 
While it is expected that {\tt Authority} will play its part, {\tt DataProvider} 
is much less than 
a peer of the contract. It is sufficient that 
it is an independent party that is entitled to call the contract's function 
to supply the expected external data. The crucial point of trust here is the 
\texttt{data\_source}, not the \texttt{DataProvider}. In other terms, since the 
parties involved in the agreement need not to trust each other, it might happen 
that {\tt DataProvider} supplies an incorrect value through the function 
\texttt{data}. In this case, the betters can appeal against the data provider since 
they agreed upon the data emitted by the \texttt{data\_source}. 
As usual, any dispute that might render the contract voidable or invalid, 
\emph{e.g.}~one better knew the result of the match in advance, can be handled by adding 
to the code of an authority party, according to the pattern illustrated in the Digital 
Licence example.

\subsection{Specification patterns in {\smallang}}

The clauses of the foregoing legal acts have been specified in {\smallang} by means of 
patterns that, in our view, are common from the juridical point of view. This is 
summarized in Table~\ref{tab.sumup}. 

{\small
\begin{table*}[hbt]
\begin{tabular}{l|l|l}
{\bf Example} & {\bf Legal clauses} & {\bf Smart Encoding in {\smallang}} \\[2mm]
\hline
Free rent & Permission and Prohibition & \begin{minipage}{7cm}
                                        States of the contract
                                       to allow or prevent the call of a function
                                       \end{minipage} 
                                        \\[4mm]
          & Obligation to return the good   & \begin{minipage}{7cm}
                                       Event: timeout that triggers a repercussion
                                      \end{minipage}\\[4mm]
          & \begin{minipage}{5cm}
             {Access to a physical good without transfer of ownership}						      
             \end{minipage}      & 
                       			\begin{minipage}{8cm}
                                (non fungible) Token: transfer only a usage code associated to the Token,   
                                i.e. operation {\small\texttt{uses(token,L)$\send$ L}}
                                \end{minipage}
               \\[4mm]
          & \begin{minipage}{5cm}
             Prohibition of preventing the usage 
             \end{minipage}   & 
                              Token held in custody in the smart contract 
                  \\[2mm]
\hline
Digital License & \begin{minipage}{5cm}
                 {Permission, Prohibition, and Obligation}						      
             \end{minipage} & States and Events
             \\[2mm]
                 & Currency and escrow    &  (divisible) Asset 
             \\[2mm]
                 & Usage and purchase & \begin{minipage}{7cm} 
                                 Token: either transfer only an associated, and personal,
                                 usage code, or transfer the token, i.e. {\small\texttt{uses(token,L)$\send$ L}} and {\small\texttt{token$\lolli$ L}}
                                 \end{minipage}\\[4mm]
                 & \begin{minipage}{5cm}
                   Off-chain constraints and 
                   Authority for dispute resolution
                   \end{minipage} & \emph{Implicitly} trusted party included in agreement
                   \\[2mm]
\hline
Bet       & Commitment and Obligation & Event\\[2mm]
          & Aleatory value & Event: timeout to decide the effective value \\[2mm]
          & Authority that decides the value & \emph{Explicitly} trusted party included in agreement
          \\[2mm]
\hline                    
\end{tabular}

\caption{\label{tab.sumup} \mbox{Legal clauses of the archetypal legal acts and their
encodings in {\smallang} }}
\end{table*}
}

\section{Related works}
\label{sec:relatedworks}

A number of projects have put forward legal markup languages, to wrap logic and 
other contextual information around traditional legal prose, and providing 
templates for common contracts. 
OpenLaw~\cite{OpenLaw} also allows to reference Ethereum-based smart contracts 
into legal agreements, and automatically trigger them once the agreement is digitally 
signed by all parties. Signatures by all relevant parties are stored on IPFS 
(the Inter-Planetary File System) and the Ethereum blockchain.
The Accord project~\cite{Accord} provides an open, standardized format for smart 
legal contracts, consisting of natural language and computable components. 
These contracts can then be interpreted by machines and they do not necessarily 
operate on  blockchains.
These projects come with sets of templates for standard legal contracts, that can 
be customized by setting template's parameters with appropriate values.
In {\smallang}, rather than software templates, it is possible to define 
 specific programming 
patterns that can be used to encode the building blocks of legal contracts.
(see the Table~\ref{tab.sumup}). 
Lexon~\cite{Lexon} uses context free grammars to define a programming language 
syntax that is at the same time human readable and automatically translated into, 
\emph{e.g.}~{\Solidity}. Even if the high level Lexon code is very close to natural 
language, there is no real control over the code that is actually run: 
the semantics of the high level language is not defined, thus the actual behaviour of 
the contract is that of the automatically generated {\Solidity} code, which might 
be much more subtle than that of the (much simpler and more abstract) Lexon source.
Compared to the {\Solidity} code of the Lexon examples in~\cite{LexonEx}, 
the {\smallang} version of the same contracts is much clearer, thanks to primitives 
like agreement and asset movements. Thus, directly coding in {\smallang} appears to 
be safer than relying on the Lexon-{\Solidity} pair. Nevertheless, it should be not difficult 
to design an automatic translation from Lexon to {\smallang} 
 so to not bind Lexon contracts to any specific implementation.

Our work aims at conducing a foundational study of legal contracts, in order to elicit 
a precisely 
defined set of building blocks that can be used to describe, analyse and execute 
(thus enforce) legal agreements. This is similar to what has been done in~\cite{PeytonJones}, 
which puts forward a set of combinators expressing financial and insurance contracts, 
together with a denotational semantics and algebraic properties that says what 
such contracts are worth. These ideas have been implemented 
by the Marlowe and Findel languages~\cite{Marlowe,Findel}, which are (small) domain specific language 
featuring constructs like participants, tokens, currency and timeouts to wait 
until a certain condition becomes true (similarly to {\smallang}). 

In Marlowe the control logic is fully embedded in the contract, that can always progress (and close) even without the participation of a non collaborative party. In particular, to ensure progress, urgency and default refund mechanism, Marlowe always impose timeouts, even for money commitments and money retrieval authorizations. However, this \emph{pull} mechanism makes the interaction logic indirect and complex, whereas in {\smallang} it is always clear who is the subject of each action and where the assets flow. Setting the correct timeouts in Marlowe may also become difficult, while in {\smallang} timeouts are optional (using events) and default mechanisms can be programmed as escape clauses that can be triggered by the Authority calling corresponding functions. Otherwise, in {\smallang}, the finite lifetime of a contract may be easily programmed  by arranging its value during the agreement phase, and issuing a corresponding event in the initial state. 
%

%
Marlowe semantics is written in Haskell and is executed on the Cardano blockchain, 
where the lifetime of a contract and timeouts are computed in terms of slot numbers,
%
%
that are measured by (Cardano's) block number. 
Then the language interpreter takes as additional input the current slot interval, 
and the contract will continue as the timeout-continuation as soon as any valid transaction is issued after the timeout's slot number is reached. 
Thus, the execution of a contract will involve multiple blocks, with multiple steps in each block. Moreover, the time limit $T$ used as timeouts must be intended as $[T-\Delta,T+\Delta]$, where $\Delta$ is a parameter of the implementation.
Similar ideas can be exploited to implement {\smallang}'s events on top of blockchains.

Overall, we remark that legal contracts are more general and more expressive than financial contracts. Accordingly, DSLs like Marlowe and Findel are built around a fixed set of contract's \emph{combinators}, that can be combined according to suitable (algebraic) laws. Then they can be implemented using an interpreter, that is a single program that handles any financial contract by evaluating its (most external) combinator. The case of {\smallang} is more complex: agreement, assets, events, named states and named functions are programming primitives rather than combinators. Therefore each {\smallang} contract must be implemented, actually compiled, into a suitable running software, e.g. a Java application or a Ethereum smart contract, and the parties must collaborate by invoking the contract's functions to make the contract progress.

Finally,
the class-based programming style of {\smallang} is similar to that 
of {\Solidity}, but there are many differences between the two. Actually,
{\smallang} is much similar to Obsidian~\cite{ObsidianPaper}, which is 
based on state-oriented programming and explicit management of linear assets, 
whose usability has been experimentally assessed~\cite{ObsidianUsage}.
Obsidian has a type system that ensures the correct manipulations of
objects
according to their current states and that linearly typed assets are not  accidentally lost. 
On the contrary, {\smallang} is untyped:
the introduction of a type discipline is orthogonal and is postponed to a 
later stage where we plan to investigate static analysis techniques specifically 
suitable to the legal setting. 
As a cons, Obsidian has no agreement nor event primitives, therefore the consensus about the contract's terms and the enforcing of legal obligations must be implemented in a much more indirect way.


\section{Conclusions}
\label{sec:conclusions}

This paper presents a domain-specific language for defining 
legal contracts that can be automatized on a blockchain system. {\smallang}
features a distilled number of operations that enable the formalisation of the main
elements of juridical acts, such as  permissions, prohibitions, and obligations. 
{\smallang} is formal and we are proud of its semantics -- the legal bisimulation --
that allows one to equate contracts that differ for clauses (events) that can never
be triggered or for the order of non-interfering communications. Furthermore,
the {\smallang} semantics has also been used for sketching an implementation 
on top of the Ethereum blockchain.

We believe that a range of legal arrangements
can be adequately translated into {\smallang},
using simple patterns for the key elements of legal contracts (see the archetypal examples in the Appendix). 
Nevertheless we acknowledge that legal contracts cannot be fully replaced by 
{\smallang} contracts, 
since the formers thoroughly use the flexibility and generality of natural languages, 
may appeal to complex and undetermined social-normative concepts 
(such as fairness or good faith), 
may need to be revised as circumstances change, may need intelligent enforcements in case on non-compliance, etc. It is matter of (our) future research to deeply investigate these
interdisciplinary aspects and provide at lest partial solutions.

From the computer science point of view, a number of issues deserve to be investigated
in full detail: the extension of the language with operations for failures, devising
linear type systems that enforce the 
partial correctness of {\smallang} codes, the implementation of the language, the 
definition of a (formal) translation in {\smallang} of a high level language, such as Lexon or part of it, in order to relieve lawyers from understanding computer science jargons. 

Overall, we are optimistic that future research on {\smallang} can  satisfactorily
address the above issues because its model is simple and rigorous, which are, in our opinion, 
fundamental criteria for reasoning about legal contracts and for 
understanding their basic principles. In our mind, {\smallang} is the backbone of a framework
where addressing and studying other, more complex features that are drawn from juridical acts.
  

%

\section*{Acknowledgements}
Giovanni Sartor has  been supported by the H2020 European Research Council (ERC) Project ``CompuLaw'' (G.A. 833647). 
Cosimo Laneve has been partly supported by the H2020-MSCA-RISE project ID 778233 ``Behavioural Application Program Interfaces (BEHAPI)''.
Silvia Crafa dedicates this work to Universit\`a di Padova and its 800th academic year.
We are grateful to the law student Alessandro Parenti for his help in constructing the legal examples in Section~\ref{sec:examples}.

\bibliographystyle{plain}
\bibliography{bibliography}

\begin{thebibliography}{10}

\bibitem{SoliditySM}
Solidity {D}ocumentation: {S}tate {M}achine {C}ommon {P}attern.
\newblock
  \url{https://docs.soliditylang.org/en/v0.8.0/common-patterns.html\#state-machine}.

\bibitem{MaltaAct}
Malta {MDIA} {A}ct.
\newblock At \url{https://mdia.gov.mt/wp-content/uploads/2018/10/MDIA.pdf},
  2018.

\bibitem{Obsidian}
Obsidian: A safer blockchain programming language.
\newblock Language Site at \url{http://obsidian-lang.com/}, 2018.

\bibitem{ItalianLaw2019}
Smart contract legislation and enforceability in {I}taly.
\newblock Gazzetta {U}fficiale, {L}aw of 11 febbraio 2019, n. 12, Art. 8 ter,
  at \url{https://www.gazzettaufficiale.it/eli/id/2019/02/12/19G00017/sg},
  2019.

\bibitem{Marlowe}
Cardano {D}ocumentation.
\newblock \url{https://docs.cardano.org/}, 2020.

\bibitem{WyomingAct}
Wyoming {R}egulation {A}ct.
\newblock At \url{https://www.wyoleg.gov/Legislation/2021/SF0038}, 2021.

\bibitem{Findel}
Alex Biryukov, Dmitry Khovratovich, and Sergei Tikhomirov.
\newblock Findel: Secure derivative contracts for ethereum.
\newblock In {\em Financial Cryptography and Data Security - {FC} 2017}, volume
  10323 of {\em Lecture Notes in Computer Science}, pages 453--467. Springer,
  2017.

\bibitem{Resources20}
Sam Blackshear, David~L. Dill, Shaz Qadeer, Clark~W. Barrett, John~C. Mitchell,
  Oded Padon, and Yoni Zohar.
\newblock Resources: {A} safe language abstraction for money.
\newblock {\em CoRR}, 2020.

\bibitem{ObsidianUsage}
Michael~J. Coblenz, Jonathan Aldrich, Brad~A. Myers, and Joshua Sunshine.
\newblock Can advanced type systems be usable? an empirical study of ownership,
  assets, and typestate in obsidian.
\newblock {\em Proc. {ACM} Program. Lang.}, 4({OOPSLA}):132:1--132:28, 2020.

\bibitem{ObsidianPaper}
Michael~J. Coblenz, Reed Oei, Tyler Etzel, Paulette Koronkevich, Miles Baker,
  Yannick Bloem, Brad~A. Myers, Joshua Sunshine, and Jonathan Aldrich.
\newblock Obsidian: Typestate and assets for safer blockchain programming.
\newblock {\em {ACM} Trans. Program. Lang. Syst.}, 42(3):14:1--14:82, 2020.

\bibitem{Nomos}
A.~Das, S.~Balzer, J.~Hoffmann, F.~Pfenning, and I.~Santurkar.
\newblock Resource-aware session types for digital contracts.
\newblock In {\em 2021 2021 IEEE 34th Computer Security Foundations Symposium
  (CSF)}, pages 111--126, Los Alamitos, CA, USA, jun 2021. IEEE Computer
  Society.

\bibitem{Move}
Sam~Blackshear et~al.
\newblock Move: A language with programmable resources.
\newblock \url{https://developers.diem.com/main/docs/move-paper}, 2021.

\bibitem{Finocchiaro2020}
Giusella Finocchiaro and Chantal Bomprezzi.
\newblock A legal analysis of the use of blockchain technology for the
  formation of smart legal contracts.
\newblock {\em MediaLaws}, July 2020.

\bibitem{ERCTokens}
Ethereum {F}oundation.
\newblock Token {S}tandards.
\newblock \url{https://ethereum.org/en/developers/docs/standards/tokens/},
  2015-21.

\bibitem{Lexon}
Lexon {F}oundation.
\newblock Lexon {H}ome {P}age.
\newblock \url{http://www.lexon.tech}, 2019.

\bibitem{LexonEx}
Lexon {F}oundation.
\newblock Lexon {D}emo {E}ditor.
\newblock \url{http://demo.lexon.tech/apps/editor/}, 2020.

\bibitem{Sartor2018}
Guido Governatori, Florian Idelberger, Zoran Milosevic, Regis Riveret, Giovanni
  Sartor, and Xiwei Xu.
\newblock On legal contracts, imperative and declarative smart contracts, and
  blockchain systems.
\newblock {\em Artificial Intelligence and Law}, 26:377--409, 2018.

\bibitem{PeytonJones}
Simon L.~Peyton Jones, Jean{-}Marc Eber, and Julian Seward.
\newblock Composing contracts: an adventure in financial engineering,
  functional pearl.
\newblock In {\em Proceedings of the Fifth {ACM} {SIGPLAN} International
  Conference on Functional Programming {(ICFP} '00), Montreal, Canada,
  September 18-21, 2000}, pages 280--292. {ACM}, 2000.

\bibitem{Lemme2019}
Giuliano Lemme.
\newblock Blockchain, smart contracts, privacy, o del nuovo manifestarsi della
  volont\`a contrattuale.
\newblock In Giuffr\`e~Francis Lefebvre, editor, {\em Privacy digitale}, pages
  293--323. Italy, 2019.

\bibitem{MilnerBook}
Robin Milner.
\newblock {\em Communication and concurrency}.
\newblock {PHI} Series in computer science. Prentice Hall, 1989.

\bibitem{Radin2017Th}
Margaret~Jane Radin.
\newblock The deformation of contract in the information society.
\newblock {\em Oxford Journal of Legal Studies}, 37, 2017.

\bibitem{Accord}
Open {S}ource {C}ontributors.
\newblock The {A}ccord {P}roject.
\newblock \url{https://accordproject.org}, 2018.

\bibitem{CommonFrameReference2009Pr}
Research Group on EC Private Law (Acquis~Group) {Study Group on a European
  Civil Code}.
\newblock {\em Principles, Definitions and Model Rules of European Private Law:
  Draft Common Frame of Reference (DCFR), Outline Edition}.
\newblock Sellier, 2009.

\bibitem{OpenLaw}
Aaron Wright, David Roon, and ConsenSys AG.
\newblock Open{L}aw {W}eb {S}ite.
\newblock \url{https://www.openlaw.io}, 2019.

\end{thebibliography}

\end{document}